\documentclass{article}  %envcountsame
\usepackage{amsmath}
\usepackage{amsfonts,amssymb,amsthm}
\usepackage{psfig}
\usepackage[all]{xy}
\usepackage{hyperref}
\usepackage{color}
\input xy
\xyoption{all}

\voffset = - 1.2cm
\theoremstyle{plain}
\newtheorem{theorem}{Theorem}
\newtheorem{corollary}[theorem]{Corollary}
\newtheorem{proposition}[theorem]{Proposition}
\newtheorem{lemma}[theorem]{Lemma}

\theoremstyle{definition}
\newtheorem{definition}[theorem]{Definition}

\newtheorem{example}[theorem]{Example}
\newtheorem{remark}[theorem]{Remark}

\numberwithin{equation}{section}
\numberwithin{theorem}{section}

\begin{document}

%\title{$k$--symplectic Lie systems: theory and applications}
\centerline{\Large \bf $k$--symplectic Lie systems: theory and applications}
%\vskip 0,25cm
\centerline{J. de Lucas$^\dagger$ and S. Vilari\~no$^\ddagger$}
\vskip 0,25cm
\centerline{$^\dagger$Department of Mathematical Methods in Physics. University of Warsaw.}

\centerline{ul. Pasteura 5, 02-093, Warszawa, Poland.}
\vskip 0,25cm
\centerline{$^\ddagger$Centro Universitario de la Defensa Zaragoza. IUMA. Universidad de Zaragoza.}

\centerline{Academia General Militar. C. de Huesca s/n, E-50090, Zaragoza, Spain. }
%\email{delucas@impan.pl}
%\email{silviavf@unizar.es}

%\titlerunning{$k$--symplectic Lie systems}

\date{Received: ---- / Accepted: ----}
%\communicated{ ----}

\begin{abstract}
A {\it Lie system} is a system of first-order ordinary differential
equations describing the integral curves of a $t$-dependent vector field taking values in a finite-dimensional real Lie algebra of vector fields: a so-called {\it Vessiot--Guldberg Lie algebra}. We suggest the definition of a particular class of Lie systems,  the $k$--symplectic Lie systems, admitting a Vessiot--Guldberg Lie algebra of Hamiltonian vector fields with respect to the presymplectic forms of a $k$--symplectic structure. We devise new $k$--symplectic geometric methods to study their
superposition rules, time independent constants of motion and general properties. Our results are
illustrated by examples of physical and mathematical interest. As a byproduct, we find a new interesting setting of application of the $k$--symplectic geometry: systems of first-order ordinary differential equations.
\end{abstract}




\section{Introduction}
The interest of geometric techniques for investigating systems of differential equations is undeniable \cite{CRC94,A89,Be84,Ol91,WTC83}. For instance, symplectic and Poisson geometry techniques have been employed to uncover interesting structures of many dynamical systems \cite{Ol91}. Further, other more recently discovered geometric structures, e.g. Dirac or Jacobi structures, have also proved their usefulness for studying differential equations and related topics \cite{Gu-1987,LZ11,YM06}. In this work, we focus upon the study of a particular class of differential equations, the Lie systems \cite{ADR12,CGM07,Dissertations,CV13,Pi13,Va13}, by means of the referred to as $k$--symplectic structures \cite{Aw-1992,Gu-1987,LV-2012}.

A {\it Lie system} is a system of first-order ordinary differential
equations  whose general solution can be expressed as a function, the {\it superposition rule}, of  a generic finite collection of particular solutions and a set of constants. In contemporary geometric terms, the {\it Lie--Scheffers Theorem} \cite{CGM07} asserts that a Lie
system is equivalent to a $t$-dependent vector field taking
values in a finite-dimensional Lie algebra of vector fields: a  {\it
Vessiot--Guldberg Lie algebra} \cite{Dissertations,Ib00,Ib09,RA}. This condition is so
stringent that just few systems of differential equations can be regarded as Lie
systems \cite{Dissertations}. Nevertheless, Lie systems appear in relevant physical and
mathematical problems and enjoy relevant geometric properties \cite{ADR12,Dissertations,
Clem06,Ru08,Ru10,Ib00,WintSecond,RA}, which strongly
prompt their analysis.

Some attention has lately been paid to Lie systems admitting
a Vessiot--Guldberg Lie algebra of Hamiltonian vector fields with respect to several geometric structures \cite{BBHLS13,BCHLS13,CLS13}. Surprisingly, studying these particular types of Lie systems led to investigate much more Lie systems and applications than before. The first attempt in this direction was performed by
Marmo, Cari\~nena and Grabowski \cite{CGM00}, who briefly studied Lie systems with Vessiot--Guldberg Lie algebras of Hamiltonian
vector fields relative to a symplectic structure. This line of research was posteriorly followed by several researchers \cite{ADR12,Ru10}.

In \cite{CLS13}, it was fully established the general theory of Lie systems admitting a Vessiot--Guldberg Lie algebra of Hamiltonian vector fields with respect to a Poisson structure, the {\it Lie--Hamilton systems}. For instance, this approach allows one to prove that the well-known invariant for Riccati equations \cite{PW}
$$
k=\frac{(x_1-x_3)(x_2-x_4)}{(x_1-x_4)(x_2-x_3)}
$$
can be retrieved as a Casimir element of a real Lie algebra of Hamiltonian functions \cite{BCHLS13}. Moreover, this work introduced the study of Poisson co-algebra techniques to obtain superposition rules for these systems \cite{BCHLS13}.

The {\it no-go Theorem for Lie--Hamilton systems}  \cite{CGLS13} is a useful tool to establish when Lie systems are not Lie--Hamilton ones. It has been employed to prove that relevant Lie systems are not Lie--Hamilton systems. Meanwhile, many such systems admit Vessiot--Guldberg Lie algebras of Hamiltonian vector fields with respect to a Dirac structure. This can be employed to
generalize the techniques employed for Lie--Hamilton systems to a larger class of Lie systems: the Dirac--Lie systems \cite{CGLS13}.

As a byproduct of studying Dirac--Lie systems, they appeared some Lie systems that admit a Vessiot--Guldberg Lie algebra with respect to several presymplectic structures. We here discover a new characteristic of many of these systems: the kernel of these presymplectic structures have trivial intersection, i.e. they form a $k$--symplectic structure \cite{Gu-1987}. Such systems are relevant as they describe Schwarzian equations \cite{CGLS13} and coupled Riccati equations \cite{BCHLS13}, which have applications in the theory of Lie systems, classical mechanics and other fields \cite{CGM07}.

In this work, we show that the above mentioned property can also be found in many other Lie systems, e.g. in Lie systems for studying diffusion equations or control systems \cite{Ni00,Ra06}. This suggests us to define a new type of Lie systems, the {\it $k$--symplectic Lie systems}, admitting Vessiot--Guldberg Lie algebras of Hamiltonian vector fields relative to the presymplectic forms of a $k$--symplectic structure.

The development of our new techniques to study $k$--symplectic Lie systems leads us to the  definition and analysis of new geometric structures for $k$--symplectic manifolds. We show that it is relevant, at least for our methods, to define generalisations of the usual structures of the symplectic geometry to the realm of $k$--symplectic structures. In this way, we define the here called $k$--symplectic Hamiltonian functions and $k$--symplectic vector fields. We construct certain Poisson algebras related to $k$--symplectic structures, the derived Poisson algebras, which are a key to obtain superposition rules for $k$--symplectic Lie systems.
This significantly improve and generalise a very few results given in \cite{Aw-1992} and \cite{merino} where analogous of our new structures, e.g. $k$--symplectic Hamiltonian functions, very briefly appear under other denominations. We have kept our terminology as we think that it reflects the fact that we are generalising presymplectic notions. Moreover, our results are more general since, for instance, the $k$--symplectic Hamiltonian functions appearing in \cite{Aw-1992}, the referred to as {\it Hamiltonian maps}, are defined only for a certain type of $k$--symplectic structures. More importantly, we show that these geometric structures play a relevant r\^ole in the theory of Lie systems, which fully justify their analysis.

At present, the $k$--symplectic geometry is mainly applied  to the study of first-order classical field theories. In particular, it allows us to give a geometric description of the Euler--Lagrange and the Hamilton--De Donder--Weyl field equations, %which are systems of partial differential equations,
as well as the study of properties of these systems such as the study of constraints, symmetries, conservation laws, reduction, etcetera \cite{Aw-1992,Gu-1987,LV-2012,MRSV13, MRS04,RSV07}. Meanwhile, we consider $k$--symplectic structures for studying systems of differential equations, which opens a new setting of application of these geometrical structures.

We demonstrate that $k$--symplectic Lie systems can be considered as Dirac--Lie systems in several non-equivalent ways. This does not mean that $k$--symplectic Lie systems must be consider simply as Dirac--Lie systems. Indeed, the techniques devised for $k$--symplectic Lie systems are more powerful since, roughly speaking, they permit us to use all these non-equivalent Dirac--Lie systems at the same time.
For instance, we illustrate that a Schwarzian equation \cite{Be07,OT09} can be studied as a $k$--symplectic Lie system or as a Dirac--Lie system in different manners. The $k$--symplectic structure allows us to obtain simultaneously several constants of motion giving rise to a superposition rule for these differential equations. Meanwhile, if we consider Schwarzian equations as Dirac--Lie system, these constants of motion must be obtained separately using different geometric arguments.

The structure of the paper goes as follows. Section \ref{LSLS} concerns the description
of the most basic notions to be used
 throughout our paper: $t$--dependent vector
  fields, Lie systems and $k$--symplectic structures. In Section \ref{NkSLS} the analysis of several remarkable Lie
systems leads us to introduce the concept of a $k$--symplectic Lie system, which encompasses such systems as
particular cases. As it can be difficult to determine whether a Lie system is a $k$--symplectic Lie system; we provide a no-go theorem to determine necessary conditions to be a $k$--symplectic Lie system in Section \ref{nogo}. Sections \ref{OHam} and \ref{DPA} are devoted to introducing some geometric structures which are employed to study $k$--symplectic Lie systems. In particular, in Section \ref{OHam} we introduce the notion of $\Omega$--Hamiltonian function as a generalisation of the Hamiltonian function notion, and in Section \ref{DPA} we relate $k$--symplectic structures to various Poisson algebras: its derived Poisson algebras. Subsequently, the $k$--symplectic Lie--Hamiltonian     structures are introduced and analyzed in Section \ref{HLS} and we analyse general properties of $k$--symplectic Lie systems in Section \ref{KSLS}. Section 9 is devoted to devising a method to calculate superposition rules for $k$--symplectic Lie systems. Finally, Section \ref{Outlook} summarizes our main results and present an
outlook of our future research on these systems.

\section{Fundamentals}\label{LSLS}

Unless otherwise stated, we assume all mathematical objects to be real, smooth,
and globally defined. This enables us to skip minor technical problems so as to stress the main points of our theory.

Given a linear space $V$ and a subset $\{v_1,\ldots,v_k\}\subset V$, we write $\langle v_1,\ldots,v_k\rangle$ for the linear hull of the vectors $v_1,\ldots,v_k$. We denote Lie algebras by pairs $(V,[\cdot,\cdot])$, where $V$ is endowed with a Lie
 bracket $[\cdot\,,\cdot]:V\times V\rightarrow V$. Given two subsets
$\mathcal{A}, \mathcal{B}\subset V$, we write $[\mathcal{A},\mathcal{B}]$ for the real linear space spanned by the Lie brackets between elements
 of $\mathcal{A}$ and $\mathcal{B}$. We define ${\rm
Lie}(\mathcal{B},V,[\cdot,\cdot])$ to be  the smallest Lie subalgebra
  of $V$ containing $\mathcal{B}$. For simplicity, we use ${\rm Lie}(\mathcal{B})$ and $V$ to represent ${\rm
Lie}(\mathcal{B},V,[\cdot,\cdot])$ and $(V,[\cdot,\cdot])$,
correspondingly, when their meaning is clear
 from context.

Let ${\rm pr}:P\rightarrow N$ be a fibre vector bundle and let $\Gamma({\rm pr})$ be
its $C^\infty(N)$--module of
  smooth sections.  If $\tau_N:TN\rightarrow N$ and $\pi_N:
T^*N\rightarrow N$ are the canonical projections
 associated with the tangent and cotangent bundle to $N$, respectively, then
$\Gamma(\tau_N)$ and $\Gamma(\pi_N)$
  designate the $C^\infty(N)$--modules of vector fields and one-forms on $N$,
correspondingly.

A {\it generalised distribution} $\mathcal{D}$ on a
manifold $N$ is a function that maps each $x\in N$ to  a linear
 subspace $\mathcal{D}_x\subset T_xN$. We say that $\mathcal{D}$ is
{\it regular at } $x'\in N$ when the function
  ${\rm r}:x\in N\mapsto \dim\mathcal{D}_x\in \mathbb{N}\cup \{0\}$  is locally constant
around $x'$. Similarly, $\mathcal{D}$ is regular on an open $U\subset N$ when ${\rm r}$ is constant on $U$. A vector field  $Y\in\Gamma(\tau_N)$ is said to
   take values in $\mathcal{D}$, in short $Y\in\mathcal{D}$, when
$Y_x\in\mathcal{D}_x$ for all $x\in N$. Likewise, similar
   notions can be defined for a {\it generalised codistribution}, namely a
mapping relating each $x\in N$ to a linear subspace of $T_x^*N$.

We call {\it $t$-dependent vector field} on $N$ a function $X:(t,x)\in\mathbb{R}\times
N\mapsto X(t,x)\in TN$
satisfying that $\tau_N\circ X=\pi_2$, for $\pi_2:(t,x)\in\mathbb{R}\times N\mapsto
x\in N$. This implies that every $t$-dependent vector field is equivalent to
a family of vector fields $\{X_t\}_{t\in\mathbb{R}}$ with $X_t:x\in N\mapsto
X(t,x)\in TN$ for all $t\in\mathbb{R}$  \cite{Dissertations}.

An {\it integral curve} of $X$ is an integral curve
$\gamma:\mathbb{R}\mapsto \mathbb{R}\times N$ of the {\it suspension} of $X$,
i.e. the vector field $\partial/\partial t+X(t,x)$ on $\mathbb{R}\times N$ \cite{FM}. Every
integral curve $\gamma$ of $X$ admits a parametrization in terms of a parameter $\bar
t$ such that $\gamma(\bar t)=(\bar t,x(\bar t))$ and
$$
\frac{{\rm d}(\pi_2 \circ \gamma)}{{\rm d}\bar t}(\bar t)=(X\circ \gamma)(\bar t).
$$
We call this system the {\it associated system} of $X$. Conversely,
every system of first-order differential equations in normal form describes the
integral curves $\gamma(\bar t)=(\bar t,x(\bar t))$ of a  unique $t$-dependent vector field. This gives rise to a
bijection between $t$-dependent vector fields and systems of first-order
differential equations in normal form, which
 justifies to use $X$ to represent both a $t$-dependent vector field and its
associated system.

\begin{definition}  The  {\it minimal Lie algebra} of a $t$-dependent vector
field $X$ on $N$ is the smallest real Lie algebra,
$V^X$, containing the vector fields $\{X_t\}_{t\in\mathbb{R}}$, namely $V^X={\rm
Lie}(\{X_t\}_{t\in\mathbb{R}})$.
\end{definition}

\begin{definition} Given a $t$-dependent vector field $X$ on $N$, its {\it
associated distribution},
$\mathcal{D}^X,$ is the generalised distribution on $N$ spanned by the vector
fields of $V^X$, i.e.
$$
\mathcal{D}^X_x=\{Y_x\mid Y\in V^X\}\subset T_xN,\qquad x\in N,
$$
and its {\it associated co-distribution}, $\mathcal{V}^X$, is the generalised
 co-distribution on $N$ of the form
 $$
\mathcal{V}^X_x=\{\theta\in T_x^*N\mid \theta(Z_x)=0,\forall
 \,\,Z_x\in \mathcal{D}_x^X\}=(\mathcal{D}^X_x)^\circ\subset T_x^*N,
 $$
 where $(\mathcal{D}^X_x)^\circ$ is the {\it annihilator} of $\mathcal{D}_x^X$.
\end{definition}

It can be proved that ${\rm r}^X:x\in N\mapsto
\dim\mathcal{D}^X_x\in\mathbb{N}\cup \{0\}$ only must be
 constant on the connected components of an open and dense subset $U^X$ of $N$
 (see \cite{CLS13}), where $\mathcal{D}^X$
  becomes a regular, involutive and integrable distribution. Since
    $\dim \mathcal{V}^X_x=\dim\, N- {\rm r}^X(x)$, then $\mathcal{V}^X$ becomes a
 regular co-distribution on each connected component of $U^X$ also.
The most relevant instance for us is when $\mathcal{D}^X$ is determined by a
finite-dimensional $V^X$. In this case, $\mathcal{D}^X$
becomes integrable on $N$ in the sense of Stefan-Sussmann \cite[p. 63]{JPOT}.
 Note that even in this case, $\mathcal{V}^X$ does not need to be a {\it
 differentiable distribution}, i.e. given
 $\theta\in\mathcal{V}^X_x$, it does not generally exist a
 locally defined one-form $\vartheta\in\mathcal{V}^X$ such that $\vartheta_x=\theta$.

Among other results, the associated distribution is important to study superposition rules for Lie systems \cite{GL13}.  Meanwhile, the associated co-distribution appears in the study of constants of motion for Lie systems \cite{CLS13}. For instance, the following proposition described in \cite{CLS13} shows that (locally defined) $t$-independent
constants of motion of $t$-dependent vector fields are determined
 by (locally defined) exact one-forms taking values in its associated
co-distribution. Then, $\mathcal{V}^X$ is what really matters in the
 calculation of such constants of motion for a system $X$.

\begin{proposition}\label{NuX} A function $f:U\rightarrow \mathbb{R}$
is a $t$-independent constant of
motion for a system $X$ on an open $U$ if and only if ${\rm d}f\in \mathcal V^X|_U$.
\end{proposition}

\begin{definition} A {\it superposition rule} depending on $m$ particular
solutions for a system $X$ on $N$
 is a mapping $\Phi:N^{m}\times N\rightarrow
N$, $x=\Phi(x_{(1)}, \ldots,x_{(m)};\lambda)$, such that the general
 solution $x(t)$ of $X$ can be cast into the form
  $x(t)=\Phi(x_{(1)}(t), \ldots,x_{(m)}(t);\lambda),$
where $x_{(1)}(t),\ldots,x_{(m)}(t)$ is any generic collection of
particular solutions and $\lambda$ is a point of $N$ to be related to initial
conditions.
 \end{definition}

The conditions ensuring that a system $X$ possesses a superposition rule are
established
 by the {\it Lie--Scheffers Theorem} \cite{CGM07,LS}.

\begin{theorem} A system $X$ admits a superposition rule if and only if $X$ can be written as
 $X={{\sum_{\alpha=1}^r}}b_\alpha(t)X_\alpha$
for a certain family $b_1(t),\ldots,b_r(t)$  of $t$-dependent functions and a
collection  $X_1,\ldots,X_r$ of vector fields  spanning
an $r$-dimensional real Lie algebra. In other words, a
system $X$ admits a superposition rule if and only if $V^X$ is
finite-dimensional.
\end{theorem}

The Lie--Scheffers Theorem may be utilised to reduce the integration  of a Lie
system to solving a special type
of Lie systems on a Lie group \cite{CGM00}. More exactly, every Lie system $X$ on a
manifold $N$ possessing a Vessiot--Guldberg
 Lie algebra $V$, let us say
$X={{\sum_{\alpha=1}^r}}b_\alpha(t)X_\alpha$, where
$X_1,\ldots,X_r$ form a basis of $V$,
  can be related to a (generally local) Lie group action $\varphi:G\times
N\rightarrow N$ whose fundamental
   vector fields coincide with those of $V$ and $\dim G=\dim V$ \cite[Theorem XI]{Palais}. This
action enables us to bring the general solution $x(t)$ of $X$ into the
   form $x(t)=\varphi(g(t),x_0)$, where $x_0\in N$ and $g(t)$ is the particular solution
with $g(0)=e$ of the Lie system
\begin{equation}\label{EquLie}
\frac{{\rm d}g}{{\rm d}t}=-\sum_{\alpha=1}^rb_\alpha(t)X_\alpha^R(g),\qquad g\in G,
\end{equation}
where $X^R_1,\ldots,X^R_r$ are a basis of the linear space of right-invariant vector fields on $G$ admitting the same structure constants as $-X_1,\ldots,-X_r$ (see \cite{CGM00}
for details). In this manner, the explicit integration of a Lie system $X$ reduces
to obtaining one particular solution of (\ref{EquLie}) if $\varphi$ is explicitly
known. Conversely, the general solution of $X$ enables us to construct the
solution
 for (\ref{EquLie}) with $g(0)=e$ by solving an algebraic system of equations obtained through $\varphi$
\cite{AHW81}.

A \textit{presymplectic manifold} is a pair $(N,\omega)$, where $\omega$ is a closed two-form on $N$.  We say that a vector field $X$ on $N$ is  \textit{Hamiltonian} with respect to $(N,\omega)$ if there exists a function $h\in \mathcal{C}^\infty(N)$ such that
$$
\iota_X\omega={\rm d}h.
$$
In this case, we call $h$ is a \textit{Hamiltonian function} for $X$. We write Adm$(\omega)$ for the space of \textit{Hamiltonian functions} relative to $(N,\omega)$. We also call these functions {\it admissible functions} of $(N,\omega)$. We hereafter denote by $X_h$, with $h\in C^\infty(N)$, a Hamiltonian vector field of $h$ relative to $\omega$. Since $\ker \omega$ may be degenerate, every function $h$ may have different Hamiltonian vector fields. It is well known that ${\rm Adm}(\omega)$ is a linear space that become a Poisson algebra when endowed with the Poisson bracket $\{\cdot,\cdot\}\colon {\rm Adm}(\omega)\times  {\rm Adm}(\omega)\to  {\rm Adm}(\omega)$ of the form
$$
\{f,g\}=X_gf,
$$
where $X_g$ is any Hamiltonian vector field of $g$. It can be proved that this definiton is independent of the chosen $X_g$ \cite{IV}.

Since $\omega$ may be degenerate, there may exist Hamiltonian vector fields related to a zero function. We call these vector fields {\it gauge vector fields} of $\omega$ and we write $G(\omega)$ for the space of such vector fields. It is immediate that $G(\omega)$ is an ideal of ${\rm Ham}(\omega)$. Hence, the space ${\rm Ham}(\omega)/G(\omega)$ is also a Lie algebra and the quotient projection $\pi:X\in {\rm Ham}(\omega)\mapsto [X]\in {\rm Ham}(\omega)/G(\omega)$ is a Lie algebra morphism. Moreover, we can define the following exact sequence of Lie algebras
$$
0\hookrightarrow {\rm H}_{\rm dR}^0(N)\hookrightarrow {\rm Adm}(\omega)\stackrel{\Lambda}{\rightarrow}\frac{{\rm Ham}(\omega)}{G(\omega)}\rightarrow0,\\
$$
where ${\rm H}_{\rm dR}^0(N)$ is the zero cohomology de Rham group of $N$ and $\Lambda:f\in {\rm Adm}(\omega)\mapsto [-X_f]\in{\rm Ham}(\omega)/G(\omega)$ (see \cite{CGLS13} for details).

\begin{definition} Let $N$ be an $n(k + 1)$-dimensional manifold and $\omega_1,\ldots,\omega_k$ a set of $k$ closed two-forms  on $N$. We say that $(\omega_1,\ldots,\omega_k)$ is a \textit{$k$--symplectic structure} if
\begin{equation}\label{k-symp_cond}
\bigcap_{i=1}^k\ker \omega_i(x)=\{0\}\,,
\end{equation}
for all $x\in N$.
We call $(N,\omega_1,\ldots,\omega_k)$ a \textit{$k$--symplectic manifold}.
\end{definition}
\begin{definition}\label{polysymplectic} A {\it $k$--polysymplectic form} on an $n(k+1)$-dimensional manifold $N$ is an $\mathbb{R}^k$-valued closed nondegenerated two-form on $N$ of the form
$$
\Omega=\sum_{i=1}^k\eta_i\otimes e^i,
$$ where
$\{e^1,\ldots,e^k\}$ is any basis for $\mathbb{R}^k$. The pair $(N,\Omega)$ is called a {\it $k$--polysymplectic manifold}.%, and the family $(N,\Omega, \mathcal{D})$ is called a standard polysymplectic manifold.
\end{definition}
\begin{remark}
 Historically, the \textit{polysymplectic structures} (see Definition \ref{polysymplectic}) were introduced by G\"{u}nther in \cite{Gu-1987}, while the notion of \textit{$k$--symplectic manifold} was introduced by Awane \cite{Aw-1992} and independently by de Le\'{o}n \textit{et al} \cite{LMS-1988, LMS-1993} under the name of \textit{$k$--cotangent structures}. Note that the notion of $k$--symplectic structure considered in this paper is not exactly the definition given by Awane, because in Awane's definition a $k$--symplectic structure on a manifold is a family of $k$ closed two-forms such that (\ref{k-symp_cond}) holds and  there exists also an integrable distribution $V$ of dimension $nk$ such that $\omega_r\vert_{V\times V}=0$ for all $r=1,\ldots, k$. Observe that when $k=1$, Awane's definition reduces to the notion of polarized symplectic manifold, that is a symplectic manifold with a Lagrangian submanifold. For that, in \cite{LV-2012} we distinguish between \textit{$k$--symplectic} and \textit{polarized $k$--symplectic manifolds} and in this paper we follows the definition of $k$--symplectic manifold considered in \cite{LV-2012}.
\end{remark}

By taking a basis $\{e^1,\ldots,e^k\}$ of $\mathbb{R}^k$, every $k$--symplectic manifold $(N,\omega_1,\ldots,\omega_k)$ gives rise to a polysymplectic manifold $(N,\Omega=\sum_{i=1}^k\omega_i\otimes e^i)$. As $\Omega$ depends on the chosen  basis, the polysymplectic manifold $(N,\Omega)$ is not canonically constructed. Nevertheless, two polysymplectic forms $\Omega_1$ and $\Omega_2$ induced by the same $k$--symplectic manifold and different bases for $\mathbb{R}^k$ are the same up to a change of basis on $\mathbb{R}^k$. In this case, we say that $\Omega_1$ and $\Omega_2$ are {\it gauge equivalent}. In a similar way, we say that $(N,\omega_1,\ldots,\omega_k)$ and $(N,\omega'_1,\ldots,\omega'_k)$ are gauge equivalent if they give rise to gauge equivalent polysymplectic forms. We can summarize these results as follows.
\begin{proposition} Let ${\rm Sym}_k(N)$ and ${\rm Pol}_k(N)$ be the spaces of $k$--symplectic and $k$--polysymplectic structures on $N$, correspondingly. The relation $(N,\omega_1,\ldots, \omega_k) \mathcal{R}_1 (N,\omega'_1,\ldots, \omega'_k)$ $(\Omega_1 \mathcal{R}_2 \Omega_2)$ if and only if the $k$--symplectic structures ($k$--polysymplectic manifolds) are gauge equivalent is an equivalence relation. Moreover,
$$
\phi: [(\omega_1,\ldots,\omega_k)]\in {\rm Sym}_k(N)/\mathcal{R}_1\mapsto \left[\sum_{i=1}^k\omega_i\otimes e^i\right]\in {\rm Pol}_k(N)/\mathcal{R}_2
$$
is a bijection.
\end{proposition}
So, we can  say that, up to gauge equivalence, $k$--symplectic and $k$--polysymplectic manifolds are essentially the same.

\begin{corollary}\label{EquMan}
Two $k$--symplectic manifolds $(N,\omega_1,\ldots,\omega_k)$ and $(N,\omega'_1,\ldots,\omega'_k)$ are equivalent if and only if $\langle \omega_1,\ldots,\omega_k\rangle=\langle \omega'_1,\ldots,\omega'_k\rangle$.
\end{corollary}

\section{On the need of $k$--symplectic Lie systems}\label{NkSLS}
In this section we show for the first time that several Lie systems of physical and mathematical interest admit Vessiot--Guldberg Lie algebras of Hamiltonian vector fields with respect to the presymplectic forms of a certain $k$--symplectic structure. This suggests us to propose a definition of $k$--symplectic Lie systems and, in following sections, to study their properties.

Consider a Schwarzian equation \cite{Be07,OT09}
\begin{equation}\label{Schwarz}
\{x,t\}=\frac{{\rm d}^3x}{{\rm d}t^3}\left(\frac{{\rm d}x}{{\rm d}t}\right)^{-1}-\frac 32\left(\frac{{\rm d}^2x}{{\rm d}t^2}\right)\left(\frac{{\rm d}x}{{\rm d}t}\right)^{-2}=2b_1(t),
\end{equation}
where $\{x,t\}$ is the refereed to as {\it Schwarzian derivative} of the function $x(t)$ in terms of the variable $t$ and $b_1(t)$ is an arbitrary $t$-dependent function. This equation is a particular case of a third-order Kummer--Schwarz equation  \cite{CGL11} and it appears in the study of iterative differential \cite{NM13}, Riccati and second-order Kummer--Schwarz equations \cite{LS13}. For simplicity, we hereafter assume $b_1(t)$ to be non-constant.

Let us analyse the properties of Schwarzian equations through a Lie system by following the exposition given in \cite{CGL11}. The first-order system of differential equations obtained by adding the variables $v\equiv {\rm d}x/{\rm d}t$ and $a\equiv {\rm d}^2x/{\rm d}t^2$ to (\ref{Schwarz}), i.e.
\begin{equation}\label{firstKS3}
\left\{\begin{aligned}
\frac{{\rm d}x}{{\rm d}t}&=v,\\
\frac{{\rm d}v}{{\rm d}t}&=a,\\
\frac{{\rm d}a}{{\rm d}t}&=\frac 32 \frac{a^2}v+2b_1(t)v,
\end{aligned}\right.
\end{equation}	
 is a Lie system. Indeed, it is the associated system to the $t$-dependent vector field
$$
X^{3KS}_t=v\frac{\partial}{\partial x}+a\frac{\partial}{\partial v}+\left(\frac 32
\frac{a^2}v+2b_1(t)v\right)\frac{\partial}{\partial
a}=Y_3+b_1(t)Y_1,
$$
where the vector fields on $\mathcal{O}_2=\{(x,v,a)\in{\rm T}^2\mathbb{R}\,|\,v\neq 0\}$, with ${\rm T^2}\mathbb{R}$ being the second tangent bundle to $\mathbb{R}$ \cite{LM87}, given by
\begin{equation}\label{VFKS1}
Y_1=2v\frac{\partial}{\partial a},\qquad Y_2=v\frac{\partial}{\partial v}+2a\frac{\partial}{\partial a},\qquad Y_3=v\frac{\partial}{\partial x}+a\frac{\partial}{\partial v}+\frac 32
\frac{a^2}v\frac{\partial}{\partial a},
\end{equation}
 satisfy
the commutation relations
\begin{equation}
 [Y_1,Y_2]=Y_1,\quad[Y_1,Y_3]=2Y_2,\quad [Y_2,Y_3]=Y_3.
\end{equation}
In consequence, $Y_1, Y_2$ and $Y_3$ span  a Lie algebra of vector fields $V^{3KS}$ isomorphic
to $\mathfrak{sl}(2,\mathbb{R})$ and $X^{3KS}$ becomes a $t$-dependent vector field taking values in $V^{3KS}$, i.e. $X^{3KS}$ is a Lie system.

Let us prove that $V^{3KS}$ is a finite-dimensional Lie algebra of Hamiltonian vector fields with respect to the presymplectic forms of a two-symplectic manifold $(\mathcal{O}_2,\omega_1,\omega_2)$. To do so, we look for presymplectic forms $\omega$ satisfying that $Y_1,Y_2$ and $Y_3$ are Hamiltonian vector fields with respect to it, i.e. $\mathcal{L}_{Y_\alpha}\omega=0$ for $\alpha=1,2,3$ and ${\rm d}\omega=0$. By solving the latter system of partial differential equations for $\omega$, we find the presymplectic forms
\begin{equation}\label{forms}
\omega_{1}\equiv \frac{{\rm d}v\wedge {\rm d}a}{v^3}, \qquad \omega_2\equiv-\frac{2}{v^3}(x\,{\rm d}v\wedge {\rm d}a+v\,{\rm d}a\wedge {\rm d}x+a\,{\rm d}x\wedge {\rm d}v).
\end{equation}

Observe that
$$
\ker\omega_1=\left\langle \frac{\partial}{\partial x}\right\rangle,\qquad \ker\omega_2=\left\langle x\frac{\partial}{\partial x}+v\frac{\partial}{\partial v}+a\frac{\partial}{\partial a}\right\rangle.
$$
Since $v\neq 0$ for every point of $\mathcal{O}_2$, then $\omega_1$ and $\omega_2$ have constant rank equal to two and $\ker \omega_1\cap\,\ker \omega_2=\{0\}$ on $\mathcal{O}_2$.
So,
$(\omega_1,\omega_2)$ forms a two-symplectic structure.

In addition, $Y_1$, $Y_2$ and $Y_3$ are Hamiltonian vector fields with respect to both presymplectic forms:
\begin{equation}\label{3KSHamFun}
\iota_{Y_1}\omega_{1}={\rm d}\left(\frac{2}{v}\right),\qquad \iota_{Y_2}\omega_{1}={\rm d}\left(\frac{a}{v^2}\right), \qquad \iota_{Y_3}\omega_{1}={\rm d}\left(\frac{a^2}{2v^3}\right),
\end{equation}
and
\begin{equation}\label{3KSHamFunOmega2}
\iota_{Y_1}\omega_2=-{\rm d}\left(\frac{4x}{v}\right),\quad
\iota_{Y_2}\omega_2={\rm d}\left(2-\frac{2ax}{v^2}\right), \quad \iota_{Y_3}\omega_2={\rm d}\left(\frac{2a}{v}-\frac{a^2x}{v^3}\right).
\end{equation}
The interest of the two-symplectic structure $(\omega_1,\omega_2)$ relies on the fact that system (\ref{firstKS3}) cannot be studied through a Lie--Hamilton system (see \cite{CGLS13} for details). Nevertheless, the use of the above presymplectic structures will allows us to
study such systems through similar techniques to those developed for Lie--Hamilton systems \cite{CGLS13}.

Let us now turn to showing that the system of Riccati equations
\begin{equation}\label{3Ricc}
\frac{{\rm d}x_{i}}{{\rm d}t}=a(t)+b(t)x_{i}+c(t)x_{i}^2, \qquad i=1,2,3,4,
\end{equation}
where $a(t),b(t)$ are arbitrary $t$-dependent functions and $c(t)$ is assumed to be different from zero,
can also be related as before to a two-symplectic structure. This system is important due to the fact that their $t$-independent constants of motion are employed to
obtain a superposition rule for $t$-dependent harmonic oscillators \cite{CLR08}. Let us show first that this system is a Lie system on $\mathcal{O}=\{(x_1,x_2,x_3,x_4)\mid \prod_{i<j}(x_i-x_j)\neq 0\}$. The system (\ref{3Ricc})
is associated to the $t$-dependent vector field
$$
X^{\rm Ric}_t=a(t)X_1+b(t)X_2+c(t)X_3,
$$
where
$$
X_1=\sum_{i=1}^4\frac{\partial}{\partial x_i},\qquad X_2=\sum_{i=1}^4x_i\frac{\partial}{\partial x_i},\qquad  X_3=\sum_{i=1}^4x_i^2\frac{\partial}{\partial x_i}.
$$
These vector fields satisfy
the commutation relations
\begin{equation}
 [X_1,X_2]=X_1,\quad[X_1,X_3]=2X_2,\quad [X_2,X_3]=X_3.
\end{equation}
Let us define the symplectic forms
$$
\omega_1=\frac{{\rm d}x_1\wedge {\rm d}x_2}{(x_1-x_2)^2}+\frac{{\rm d}x_3\wedge {\rm d}x_4}{(x_3-x_4)^2},\qquad
\omega_2=\sum_{\stackrel{i,j=1}{i< j}}^4 \frac{{\rm d}x_i\wedge {\rm d}x_j}{(x_i-x_j)^2}.
$$
Hence, $(\omega_1,\omega_2)$ becomes a two-symplectic structure.
Additionally, the vector fields $X_1,X_2$ and $X_3$ are Hamiltonian relative to $\omega_1$ and $\omega_2$:
\begin{equation}\label{RHamFun}\begin{gathered}
\iota_{X_1}\omega_1\!=\!{\rm d}\Big(\frac{1}{x_1-x_2}+\frac{1}{x_3-x_4}\Big), \qquad
\iota_{X_2}\omega_1\!=\!
\frac 12{\rm d}\Big(\frac{x_1+x_2}{x_1-x_2}\!+\!\frac{x_3+x_4}{x_3-x_4}\Big), \,\,\\\noalign{\medskip}
\iota_{X_3}\omega_1\!=\!{\rm d}\Big(\frac{x_1 x_2}{x_1-x_2}\!+\!\frac{x_3 x_4}{x_3-x_4}\Big)
\end{gathered}
\end{equation}
and
\begin{equation}\label{RHamFunOmega2}
\iota_{X_1}\omega_2={\rm d}\left(\sum_{\stackrel{i,j=1}{i<j}}^4\frac{1}{x_i-x_j}\right),\qquad \iota_{X_2}\omega_2=\frac 12{\rm d}\left(\sum_{\stackrel{i,j=1}{i< j}}^4\frac{x_i+x_j}{x_i-x_j}\right),\qquad \iota_{X_3}\omega_2=\!{\rm d}\left(\sum_{\stackrel{i,j=1}{i< j}}^4\frac{x_i x_j}{x_i-x_j}\right).
\end{equation}

Let us now turn to the system of differential equations
$$
\begin{gathered}
\frac{{\rm d}x_1}{{\rm d}t}=b_1(t),\qquad \frac{{\rm d}x_2}{{\rm d}t}=b_2(t),\qquad
\frac{{\rm d}x_3}{{\rm d}t}=b_2(t)x_1,\qquad \frac{{\rm d}x_4}{{\rm d}t}=b_2(t)x_1^2, \qquad \frac{{\rm d}x_5}{{\rm d}t}=2b_2(t)x_1x_2,\
\end{gathered}
$$
where $b_1(t)$ and $b_2(t)$ are arbitrary $t$-dependent functions and whose interest is due to its relation to certain control problems \cite{Ni00,Ra06}. This system is associated to the $t$-dependent vector field $X_t=b_1(t)X_1+b_2(t)X_2,$
with
$$
X_1=\frac{\partial}{\partial x_1},\quad X_2=\frac{\partial}{\partial x_2}+x_1\frac{\partial}{\partial x_3}+x_1^2\frac{\partial}{\partial x_4}+2x_1x_2\frac{\partial}{\partial x_5}.
$$
These vector fields span a Lie algebra $V$ of vector fields along with
$$
X_3=\frac{\partial}{\partial x_3}+2x_1\frac{\partial}{\partial x_4}+2x_2\frac{\partial}{\partial x_5},\qquad X_4=\frac{\partial}{\partial x_4},\qquad X_5=\frac{\partial}{\partial x_5}.
$$
Indeed, the only non-vanishing commutation relations between the previous vector fields read
\[
[X_1,X_2]=X_3,\qquad [X_1,X_3]=2X_4,\qquad [X_2,X_3]=2X_5.
\]
Consequently, $X$ is a Lie system as indicated in \cite{Ra06}. Additionally to this, we can add that the Lie algebra $V$ consists of Hamiltonian vector fields relative to the presymplectic forms
$$
\omega_1={\rm d}x_1\wedge {\rm d}x_2,\qquad \omega_2={\rm d}x_1\wedge {\rm d}x_3,\qquad \omega_3={\rm d}x_1\wedge {\rm d}x_4,\qquad \omega_4={\rm d}x_2\wedge {\rm d}x_5+x_2^2{\rm d}x_1\wedge {\rm d}x_2.
$$
The kernels of the above presymplectic forms are
$$\begin{gathered}
\ker \omega_{1}=\left\langle \frac{\partial}{\partial x_3},\frac{\partial}{\partial x_4},\frac{\partial}{\partial x_5}\right\rangle,\qquad \ker\omega_2=\left\langle \frac{\partial}{\partial x_2},\frac{\partial}{\partial x_4},\frac{\partial}{\partial x_5}\right\rangle,\qquad
\ker\omega_3=\left\langle \frac{\partial}{\partial x_2},\frac{\partial}{\partial x_3},\frac{\partial}{\partial x_5}\right\rangle,\\
\ker \omega_{4}=\left\langle \frac{\partial}{\partial x_3},\frac{\partial}{\partial x_4},\frac{\partial}{\partial x_1}+x_2^2\frac{\partial}{\partial x_5}\right\rangle.
\end{gathered}
$$
Obviously, $\cap_{i=1}^4\ker\omega_i=\{0\}$ and $(\omega_1,\ldots,\omega_4)$ become a 4-symplectic structure. In addition, $X_1,X_2,X_3,$ $X_4$ and $X_5$ are Hamiltonian vector fields with respect to these presymplectic forms. In fact,
\[
\begin{array}{lcllcllcllcl}
\iota_{X_1}\omega_1 & \!=\! & {\rm d}x_ 2, & \iota_{X_1}\omega_2& \!=\! & {\rm d}x_3 ,& \iota_{X_1}\omega_3& \!=\! & {\rm d}x_4 ,& \iota_{X_1}\omega_4& \!=\! & \frac{1}{3}{\rm d}x_2^3, \\\noalign{\medskip}
\iota_{X_2}\omega_1 & \!=\! & -{\rm d}x_ 1, & \iota_{X_2}\omega_2& \!=\! & -\frac{1}{2}{\rm d}x_1^2 ,& \iota_{X_2}\omega_3& \!=\! & -\frac{1}{3}{\rm d}x_1^3 ,& \iota_{X_2}\omega_4& \!=\! & {\rm d}(x_5-x_1x_2^2), \\\noalign{\medskip}
\iota_{X_3}\omega_1 & \!=\! & 0, & \iota_{X_3}\omega_2& \!=\! & {\rm d}x_1 ,& \iota_{X_3}\omega_3& \!=\! & -{\rm d}x_1^2 ,& \iota_{X_3}\omega_4& \!=\! & -{\rm d}x_2^2, \\\noalign{\medskip}
\iota_{X_4}\omega_1 & \!=\! & 0, & \iota_{X_4}\omega_2& \!=\! & 0 ,& \iota_{X_4}\omega_3& \!=\! & -{\rm d}x_1 ,& \iota_{X_4}\omega_4& \!=\! & 0, \\\noalign{\medskip}
\iota_{X_5}\omega_1 & \!=\! & 0, & \iota_{X_5}\omega_2& \!=\! & 0 ,& \iota_{X_5}\omega_3& \!=\! &0 ,& \iota_{X_5}\omega_4& \!=\! & -{\rm d}x_2\,. \\\noalign{\medskip}
\end{array}
\]

Let us now consider the control system in $\mathbb{R}^5$ \cite{Ra06}
$$
\begin{gathered}
\frac{{\rm d}x_1}{{\rm d}t}=b_1(t),\quad
\frac{{\rm d}x_2}{{\rm d}t}=b_2(t),\quad
\frac{{\rm d}x_3}{{\rm d}t}=b_2(t)x_1-b_1(t)x_2 ,\quad
\frac{{\rm d}x_4}{{\rm d}t}=b_2(t)x_1^2,\quad
\frac{{\rm d}x_5}{{\rm d}t}=b_1(t)x_2^2.
\end{gathered}
$$
 This system is associated to the $t$-dependent vector field
$$
X_t=b_1(t)X_1+b_2(t)X_2,
$$
with
$$
X_1=\frac{\partial}{\partial x_1}-x_2\frac{\partial}{\partial x_3}+ x_2^2\frac{\partial}{\partial x_5},\qquad X_2=\frac{\partial}{\partial x_2}+x_1\frac{\partial}{\partial x_3}+x_1^2\frac{\partial}{\partial x_4}.
$$
These vector fields span a Lie algebra $V$ of vector fields along with
$$
X_3=\frac{\partial}{\partial x_3}+x_1\frac{\partial}{\partial x_4}-x_2\frac{\partial}{\partial x_5},\qquad X_4=\frac{\partial}{\partial x_4},\qquad X_5=\frac{\partial}{\partial x_5}.
$$
Indeed, the only non-vanishing commutation relations between the previous vector fields read
$$
[X_1,X_2]=2X_3,\qquad [X_1,X_3]=X_4,\qquad [X_2,X_3]=-X_5.
$$
Consequently, $X$ is a Lie system. Additionally, the Lie algebra $V$ consists of Hamiltonian vector fields relative to the presymplectic forms
$$
\omega_1={\rm d}x_1\wedge {\rm d}x_2,\qquad \omega_2={\rm d}x_2\wedge {\rm d}x_5,\qquad \omega_3={\rm d}x_1\wedge {\rm d}x_4,\qquad \omega_4={\rm d}x_1\wedge {\rm d}x_3+x_1{\rm d}x_1\wedge {\rm d}x_2.
$$
The kernels of the above presymplectic forms read
$$\begin{gathered}
\ker \omega_{1}=\left\langle \frac{\partial}{\partial x_3},\frac{\partial}{\partial x_4},\frac{\partial}{\partial x_5}\right\rangle,\qquad \ker\omega_2=\left\langle \frac{\partial}{\partial x_1},\frac{\partial}{\partial x_3},\frac{\partial}{\partial x_4}\right\rangle,\qquad
\ker\omega_3=\left\langle \frac{\partial}{\partial x_2},\frac{\partial}{\partial x_3},\frac{\partial}{\partial x_5}\right\rangle,\\
\ker \omega_{4}=\left\langle \frac{\partial}{\partial x_4},\frac{\partial}{\partial x_5},\frac{\partial}{\partial x_2}-x_1\frac{\partial}{\partial x_3}\right\rangle.
\end{gathered}
$$
Obviously, $\cap_{i=1}^4\ker\omega_i=\{0\}$ and $(\omega_1,\ldots,\omega_4)$ become a 4-symplectic structure. In addition,  $X_1,X_2,X_3,$ $X_4$ and $X_5$ are Hamiltonian vector fields with respect to the four presymplectic forms. In fact,
\[
\begin{array}{lcllcllcllcl}
\iota_{X_1}\omega_1 & = & {\rm d}x_ 2, & \iota_{X_1}\omega_2& = & -\frac{1}{3}{\rm d}x_2^3 ,& \iota_{X_1}\omega_3& = & {\rm d}x_4 ,& \iota_{X_1}\omega_4& = & {\rm d}(x_1x_2+x_3) , \\\noalign{\medskip}
\iota_{X_2}\omega_1 & = & -{\rm d}x_ 1, & \iota_{X_2}\omega_2& = & {\rm d}x_5 ,& \iota_{X_2}\omega_3& = & -\frac{1}{3}{\rm d}x_1^3 ,& \iota_{X_2}\omega_4& = & -{\rm d}x_1^2, \\\noalign{\medskip}
\iota_{X_3}\omega_1 & = & 0, & \iota_{X_3}\omega_2& = & \frac{1}{2} {\rm d}x_2^2 ,& \iota_{X_3}\omega_3& = & -\frac 12{\rm d}x_1^2 ,& \iota_{X_3}\omega_4& = & -{\rm d}x_1, \\\noalign{\medskip}
\iota_{X_4}\omega_1 & = & 0, & \iota_{X_4}\omega_2& = & 0 ,& \iota_{X_4}\omega_3& = & -{\rm d}x_1 ,& \iota_{X_4}\omega_4& = & 0, \\\noalign{\medskip}
\iota_{X_5}\omega_1 & = & 0, & \iota_{X_5}\omega_2& = & -{\rm d}x_2 ,& \iota_{X_5}\omega_3& = &0 ,& \iota_{X_5}\omega_4& = & 0. \\\noalign{\medskip}
\end{array}
\]

It was recently proved that diffusion equations and other PDEs can be approached through the Lie system
\begin{equation*}
\left\{
\begin{aligned}
\frac{{\rm d}s}{{\rm d}t}&=-4a(t)u s-2d(t)s,\qquad &\frac{{\rm d}x}{{\rm d}t}&=(c(t)+4a(t)u)x+f(t)-2u g(t), \\
\frac{{\rm d}u}{{\rm d}t}&=-b(t)+2c(t) u+4a(t)u^2,\qquad &\frac{{\rm d}y}{{\rm d}t}&=(2a(t)x-g(t))v,\\
\frac{{\rm d}v}{{\rm d}t}&=(c(t)+4a(t)u)v,\qquad &\frac{{\rm d}z}{{\rm d}t}&=a(t)x^2-g(t)x, \\
\frac{{\rm d}w}{{\rm d}t}&=a(t)v^2, &\\
\end{aligned}\right.
\end{equation*}
where $a(t),b(t),c(t),d(t),f(t),g(t)$ are certain $t$-dependent functions (see \cite{CGLS13,LLS11,SSV11} for details).
Its general solution can be obtained by integrating
\begin{equation}\label{Partial}
\left\{
\begin{aligned}
\frac{{\rm d}u}{{\rm d}t}&=-b(t)+2c(t) u+4a(t)u^2,\\
\frac{{\rm d}v}{{\rm d}t}&=(c(t)+4a(t)u)v,\\
\frac{{\rm d}w}{{\rm d}t}&=a(t)v^2.\\
\end{aligned}\right.
\end{equation}
This is a Lie system \cite{CGLS13}. In fact, it describes the integral curves of the $t$-dependent vector field
$$
X^{RS}_t=a(t)X_1-b(t)X_2+c(t)X_3,
$$
where
$$
\begin{gathered}
 X_1=4u^2\frac{\partial}{\partial u}+4uv\frac{\partial}{\partial v}+v^2\frac{\partial}{\partial w},\quad  X_2=\frac{\partial}{\partial u},\quad X_3=2u\frac{\partial}{\partial u}+v\frac{\partial }{\partial  v}
\end{gathered}
$$
close on the commutation relations
$$
\begin{gathered}
\left[X_1,X_2\right]=-4X_3, \qquad [X_1,X_3]=-2X_1,\qquad [X_2,X_3]=2X_2.
\end{gathered}
$$
Consider the presymplectic forms
\begin{equation}\label{RS-1Sym}
\omega_{RS-1}=-\frac{4w{\rm d}u\wedge {\rm d}w}{v^2}+\frac{{\rm d}v\wedge {\rm d}w}{v}+\frac{4w^2{\rm d}u\wedge {\rm d}v}{v^3},\qquad \omega_{RS-2}=-\frac{4{\rm d}u\wedge {\rm d}w}{v^2}+\frac{8w{\rm d}u\wedge {\rm d}v}{v^3}.
\end{equation}
Their kernels read
$$
\ker \omega_{RS-1}=\left\langle v^2\frac{\partial}{\partial u}+4wv\frac{\partial}{\partial v}+4w^2\frac{\partial}{\partial w}\right\rangle,\qquad
\ker \omega_{RS-2}=\left\langle v\frac{\partial}{\partial v}+2w\frac{\partial}{\partial w}\right\rangle.
$$
Note that $\ker \omega_{RS-1}\cap \ker \omega_{RS-2}=\{0\}$% and $\omega_{RS-1}|_V=\omega_{RS-2}|_V$ with $V=\ker \omega_{RS-1}\oplus\omega_{RS-2}$
. So,
$(\omega_{RS-1},\omega_{RS-2})$ is a two-symplectic structure. Moreover, $X_1$, $X_2$ and $X_3$ are Hamiltonian vector fields with respect to $\omega_{RS-1}$, $\omega_{RS-2}$:
\begin{equation}\label{eq1}
\iota_{X_1}\omega_{RS-1}=\!{\rm d}\Big(4uw-8\frac{u^2w^2}{v^2}-\frac{v^2}{2}\Big), \
\iota_{X_2}\omega_{RS-1}=\! -2
{\rm d}\Big(\frac{w^2}{v^2}\Big), \
\iota_{X_3}\omega_{RS-1}=\!{\rm d}\Big(w-4\frac{uw^2}{v^2}\Big),
\end{equation}
and
\begin{equation}\label{eq2}
\iota_{X_1}\omega_{RS-2}=4\,{\rm d}\Big(u-4\frac{u^2w}{v^2}\Big), \
\iota_{X_2}\omega_{RS-2}= -4\,
{\rm d}\Big(\frac{w}{v^2}\Big), \
\iota_{X_3}\omega_{RS-2}=-8\,{\rm d}\Big(\frac{uw}{v^2}\Big).
\end{equation}

 Let us consider a type of Lotka--Volterra systems, i.e. a system of the form
$$
\frac{{\rm d}x_i}{{\rm d}t}=x_i\Big(b_i(t)+\sum_{j=1}^nc_{ij}(t)x_j\Big),\qquad i=1,\ldots,n,
$$
for certain functions $b_i(t)$ and $c_{ij}(t)$, that can be studied as a Lie system, namely its minimal Lie algebra is finite-dimensional. Systems of this type have already been studied by one  of the authors of this work in \cite{BBHLS13}. We hereafter call these systems {\it Lie--Lotka--Volterra systems}. More specifically, consider the system
\begin{equation}\label{LVsystem}
    \begin{gathered}
\frac{{\rm d}x_1}{{\rm d}t}=a(t)x_1+b(t)x_1^2,\qquad
\frac{{\rm d}x_2}{{\rm d}t}=a(t)x_2+b(t)x_2^2,\qquad
\frac{{\rm d}x_3}{{\rm d}t}=a(t)x_3+b(t)x_3^2,\\
\frac{{\rm d}x_4}{{\rm d}t}=a(t)x_4+b(t)x_4^2,\qquad
\frac{{\rm d}x_5}{{\rm d}t}=a(t)x_5+b(t)x_5^2.
    \end{gathered}
\end{equation}
%which can also be considered as a matrix Riccati equation \cite{SW08} or a matrix Bernoulli equation \cite{De08}.
Observe that its associated $t$-dependent vector field is of the form
$$
X=a(t)X_1+b(t)X_2,
$$
where
$$
X_1=\sum_{i=1}^5x_i\frac{\partial}{\partial x_i},\qquad X_2=\sum_{i=1}^5x_i^2\frac{\partial}{\partial x_i}
$$
satisfy $[X_1,X_2]=X_2$. So, $X$ admits a Vessiot--Guldberg Lie algebra isomorphic to
the Lie algebra of affine transformations on the real line. Let us show that this system is a four-symplectic Lie system. Consider the presymplectic forms
\begin{equation}\label{LV_4_symp}
    \begin{gathered}
    \omega_1=\frac{{\rm d}x_1\wedge {\rm d}x_2}{(x_1-x_2)^2} +\frac{{\rm d}x_3\wedge {\rm d}x_4}{(x_3-x_4)^2},\qquad
    \omega_2=\frac{{\rm d}x_1\wedge {\rm d}x_2}{(x_1-x_2)^2} +\frac{{\rm d}x_3\wedge {\rm d}x_5}{(x_3-x_5)^2},\\
    \omega_3=\frac{{\rm d}x_1\wedge {\rm d}x_2}{(x_1-x_2)^2} +\frac{{\rm d}x_4\wedge {\rm d}x_5}{(x_4-x_5)^2},\qquad
    \omega_4=\frac{{\rm d}x_1\wedge {\rm d}x_3}{(x_1-x_3)^2} +\frac{{\rm d}x_4\wedge {\rm d}x_5}{(x_4-x_5)^2}.
    \end{gathered}
\end{equation}

Note that $\cap_{i=1}^4\ker\omega_i=\{0\}$. So $(\omega_1,\omega_2,\omega_3,\omega_4)$ is a four-symplectic structure. Moreover $X_1$ and $X_2$ are Hamiltonian vector fields with respect to $\omega_i$ with $ i=1,\ldots, 4$:
\begin{equation}\label{LVs_Ham1}
    \begin{gathered}
        \iota_{X_1}\omega_1=\frac{1}{2}{\rm d}\Big(\frac{x_1+x_2}{x_1-x_2} + \frac{x_3+x_4}{x_3-x_4} \Big),\qquad \iota_{X_2}\omega_1={\rm d}\Big(\frac{x_1x_2}{x_1-x_2}+ \frac{x_3x_4}{x_3-x_4} \Big),\\
        \iota_{X_1}\omega_2=\frac{1}{2}{\rm d}\Big(\frac{x_1+x_2}{x_1-x_2} + \frac{x_3+x_5}{x_3-x_5} \Big),\qquad \iota_{X_2}\omega_2={\rm d}\Big(\frac{x_1x_2}{x_1-x_2}+ \frac{x_3x_5}{x_3-x_5} \Big),\\
        \iota_{X_1}\omega_3=\frac{1}{2}{\rm d}\Big(\frac{x_1+x_2}{x_1-x_2} + \frac{x_4+x_5}{x_4-x_5} \Big),\qquad \iota_{X_2}\omega_3={\rm d}\Big(\frac{x_1x_2}{x_1-x_2}+ \frac{x_4x_5}{x_4-x_5} \Big),\\
        \iota_{X_1}\omega_4=\frac{1}{2}{\rm d}\Big(\frac{x_1+x_3}{x_1-x_3} + \frac{x_4+x_5}{x_4-x_5} \Big),\qquad \iota_{X_2}\omega_4={\rm d}\Big(\frac{x_1x_3}{x_1-x_3}+ \frac{x_4x_5}{x_4-x_5} \Big)\,.
    \end{gathered}
\end{equation}

All above systems posses a Lie algebra of vector fields that are Hamiltonian with respect to all the presymplectic forms belonging to a $k$--symplectic structure, this suggests us the following definition.

\begin{definition}
Given a $k$--symplectic structure $(\omega_1,\ldots,\omega_k)$ on an $n(k+1)$ dimensional manifold $N$, we say that a vector field $Y$ on $N$ is {\it $k$--Hamiltonian} if it is a Hamiltonian vector field with respect to the presymplectic forms $\omega_1,\ldots,\omega_k$.
\end{definition}

Note that $X$ is a $k$--Hamiltonian vector field if and only if it is Hamiltonian for all the presymplectic forms of the space $\langle \omega_1,\ldots,\omega_k\rangle$. In view of Theorem \ref{EquMan}, two $k$--symplectic structures $(\omega_1,\ldots,\omega_k)$ and $(\omega'_1,\ldots,\omega_k')$ are equivalent if and only if $\langle \omega_1,\ldots,\omega_k\rangle=\langle \omega'_1,\ldots,\omega'_k\rangle$. So, if $X$ is $k$--Hamiltonian for a $k$--symplectic manifold, it is $k$--Hamiltonian for all equivalent $k$--symplectic manifolds. In addition, it also makes sense to say that $X$ is $\Omega$--Hamiltonian for a polysymplectic form $\Omega$ if $X$ is $k$--Hamiltonian for a $k$--symplectic manifold possessing $\Omega$ as associated polysymplectic form. From now on, we will talk about \textit{$k$--Hamiltonian} and/or \textit{$\Omega$--Hamiltonian vector fields} indistinctly.  We write Ham$(\Omega)$, where $\Omega$ is a polysymplectic form induced by $(\omega_1,\ldots,\omega_k)$, for the space of $k$--Hamiltonian vector fields.

Now, it makes sense to define the following notion of $k$--symplectic Lie systems.
\begin{definition}
We say that a system $X$ is a {\it $k$--symplectic Lie system} if $V^X$ is a
finite-dimensional real Lie algebra of $k$--Hamiltonian vector fields with respect to
a $k$--symplectic structure $(\omega_1,\ldots,\omega_k)$. We call $(\omega_1,\ldots,\omega_k)$ a {\it compatible} $k$--symplectic structure.
\end{definition}

Note that the above can be restated by saying that a system  $X$ on a manifold $N$ is a $k$--symplectic Lie system if and only if it admits a Vessiot--Guldberg Lie algebra of $k$--Hamiltonian vector fields with respect to a certain $k$--symplectic structure on $N$. Observe that Lie--Hamilton systems related to symplectic structures \cite{CLS13} are a particular type of $k$--symplectic Lie systems. Nevertheless, we already commented that not every $k$--symplectic Lie system is a Lie--Hamilton system (for more details see description of the system (\ref{firstKS3}) and \cite{CGLS13}).

Every $k$--symplectic Lie system can be considered as a Dirac--Lie system \cite{CGLS13}.  More specifically, if $X$ is a $k$--symplectic Lie system relative to the $k$--symplectic structure $(\omega_1,\ldots,\omega_k)$, then $V^X$ is a family of Hamiltonian vector fields with respect to each one of the presymplectic forms $\omega_1,\ldots,\omega_k$. So, $V^X$ is a Lie algebra of Hamiltonian vector fields relative to each Dirac structure $L^{\omega_r}$ induced by the presymplectic form $\omega_r$ (see \cite{CGLS13,Co90} for details). Following the notation of \cite{CGLS13}, we say that the triple $(N,L^{\omega_r},X)$ is a Dirac--Lie system. Moreover, it can be proved that every Lie system can be considered as a Dirac--Lie system \cite{CGLS13}.
Meanwhile, not every Dirac--Lie system can be considered as a $k$--symplectic Lie system, e.g. a Lie system given by an autonomous vector field $X\neq 0$ on the real line. Nevertheless, the main advantage of $k$--symplectic Lie systems is that they can be considered as Dirac--Lie systems in
different ways. This suggests us to find a natural approach to the study of these systems, which is given by $k$--symplectic structures.

\section{A no-go theorem for $k$--symplectic Lie systems}\label{nogo}
Determining whether a Lie system is a $k$--symplectic Lie system generally requires  solving a system of PDEs to find a compatible $k$--symplectic structure. In many cases, it can be difficult to establish whether this system of PDEs has enough solutions giving rise to a compatible $k$--symplectic structure. That is why it is important to find simple necessary and/or sufficient conditions to ensure or to discard that a Lie system is a  $k$--symplectic Lie system.

In this section we provide a no-go theorem giving conditions ensuring that a Lie system is not a $k$--symplectic Lie system. The main idea is that the minimal Lie algebra of the Lie system under study must leave stable, in the sense given next, the kernels of the presymplectic forms of any $k$--symplectic structure compatible with the Lie system. This condition is easier to verify than finding a compatible $k$--symplectic structure. Although we here provide only one main result, it is easy to develop further no-go theorems from our ideas.

\begin{definition} We say that a distribution $\mathcal{D}$ is {\it stable} under the action of a Lie algebra $V$ of vector fields when $[X,Y]\in \mathcal{D}$ for every $Y\in\mathcal{D}$ and $X\in V$.
\end{definition}
\begin{definition} Given a finite-dimensional real Lie algebra $V$ of vector fields on $N$, we say that  $V$ is $s$--{\it primitive} when there exists no distribution $\mathcal{D}$ of rank $s$ stable under the action of $V$.  We call $V$ {\it odd--primitive} when $V$ is $s$--primitive for every odd value of $s<\dim N$.

\end{definition}
\begin{remark} Observe that the above definition of $s$-primitive Lie algebra of vector fields is a generalisation of the notion of a primitive Lie algebra of vector fields on the plane given in \cite{GKO92}.
\end{remark}

\begin{theorem}\label{Nogo} {\bf (No-go $k$--symplectic Lie systems theorem)} If $X$ is a Lie system on an odd dimensional manifold $N$
and $V^X$ is odd-primitive, then $X$ is not
a $k$--symplectic Lie system.
\end{theorem}
\begin{proof} Let us suppose that there exists a compatible $k$--symplectic structure $(\omega_1,\ldots,\omega_k)$ for $X$. On an odd-dimensional manifold, every two-form of the $k$--symplectic structure has non-trivial odd-dimensional kernel. Let $Z\neq 0$ be a vector field $Z\in\ker \omega_i$. As the elements of $V^X$ are Hamiltonian with respect to each one of the two-forms of the $k$--symplectic structure, we have that $\mathcal{L}_Y\omega_i=0$ for every $Y\in V^X$ and
$$
0=\mathcal{L}_Y\iota_Z\omega_i=\iota_Z\mathcal{L}_Y\omega_i+\iota_{[Y,Z]}\omega_i=\iota_{[Y,Z]}\omega_i.
$$
So $[Y,Z]\in \ker \omega_i$ and the kernel of $\omega_i$ is stable under the action of the elements of $V^X$. As $V^X$ is odd-primitive and $\ker\,\omega_i$ is odd-dimensional, this is a contradiction. Then, the compatible $k$--symplectic structure cannot exist.
\end{proof}

\begin{example}({\bf Lie systems on Lie groups}) Let us consider the following type of Lie systems on Lie groups  of the form
\begin{equation}\label{Lie}
\frac{{\rm d}g}{{\rm d}t}=\sum_{\alpha=1}^rb^R_\alpha(t)X^R_\alpha(g)+\sum_{\alpha=1}^rb^L_\alpha(t)X^L_\alpha(g), \qquad g\in G,
\end{equation}
where $G$ is a Lie group, $X^R_1,\ldots,X^R_r$  and $X^L_1,\ldots,X^L_r$ form basis of right and left-invariant vector fields on $G$ respectively,  and $b^L_1(t),\ldots,b^L_r(t),b^R_1(t),\ldots,b^R_r(t)$ are arbitrary $t$-dependent functions. Additionally, we assume $G$ to be connected. Systems of the type (\ref{Lie}) appear when searching for transformations mapping a Lie system into a new one, e.g. in a reduction process \cite{IntRic11}. Additionally, each Lie system on a manifold can be solved by means of a particular solution of systems like (\ref{Lie}) where only right-invariant or left-invariant vector fields appear. Moreover,  such systems  appear in Control Theory and Darboux integrable systems \cite{CCR03,Va13}.
An interesting question is to determine if such systems can be endowed with a compatible $k$--symplectic structure. As proved next, the answer is negative for a large family of systems (\ref{Lie}).

Assume that (\ref{Lie}) is such that $G$ is odd-dimensional and its Lie algebra of left-invariant vector fields, $\mathfrak{g}$, has no odd-dimensional ideals, e.g. $\mathfrak{g}$ is simple. Suppose also that the minimal Lie algebra $V^X$ of the system (\ref{Lie}) is isomorphic (as a Lie algebra) to $\mathfrak{g}\oplus\mathfrak{g}$, namely $V^X=\langle X^R_1,\ldots,X^R_r,X^L_1,\ldots,X^L_r\rangle$. Let us prove that the $V^X$ is odd-primitive by reduction to absurd.

Consider an odd-dimensional distribution $\mathcal{D}$ on $G$ invariant under the action of $V^X$. Since $G$ is connected, the invariance of $\mathcal{D}$ implies that the vector fields taking values in $\mathcal{D}$ are invariant under the diffeomorphisms given by the left-hand multiplications on the group, namely the mappings $L_g:g'\in G\mapsto g\cdot g'\in G$ with $g\in G$. So, given $\mathcal{D}_e\subset T_eG$, i.e. the subspace of the distribution $\mathcal{D}$ at the neutral element $e$ of $G$, we obtain that $(L_{g})_*Y\in \mathcal{D}_g$ for every $Y\in \mathcal{D}_e$. Indeed, since $L_g$ is a diffeomorphism, then  $\mathcal{D}_e\simeq \mathcal{D}_g$.

If $Y^L$ is a left-invariant vector field on $G$ with $Y^L(e)=Y$, we have that
$$
Y^L(g)=(L_g)_*Y^L(e)\in \mathcal{D}_g$$
for every $g\in G$. So, $Y^L$ takes values in $\mathcal{D}$. Using that each $(L_g)_*$ is a diffeomorphism, we obtain that given a set of vector fields $Y^L_1,\ldots, Y^L_s$ whose values $Y^L_1(e),\ldots, Y^L_s(e)$ form a basis for $\mathcal{D}_e$, then $(Y^L_1)(g),\ldots, (Y^L_s)(g)$ form a basis of $\mathcal{D}_g$ for every $g\in G$. Consequently, $\mathcal{D}$ admits a global basis of left-invariant vector fields $Y^L_1,\ldots,Y^L_s$. As $\mathcal{D}$ is invariant under $V^X$, then $\mathcal{L}_{X^L_\alpha}Y^L_j$, with $\alpha=1,\ldots,\dim G$ and $j=1,\ldots,s$, is a right-invariant vector field taking values in $\mathcal{D}$. Hence,
$$[X_\alpha^L,Y_j^L]\in \langle Y_1^L\ldots,Y^L_s\rangle$$
and $\langle Y^L_1,\ldots,Y^L_s\rangle$ is an odd-dimensional ideal of $\mathfrak{g}$. By assumption, $\mathfrak{g}$ has no odd-dimensional ideals. This is a contradiction and we have that $\mathcal{D}$ is not invariant under $V^X$. By Theorem \ref{Nogo}, we obtain that (\ref{Lie}) is not a $k$--symplectic Lie system.
\end{example}
\section{On {\rm $\Omega$}--Hamiltonian functions}\label{OHam}

Every $k$--Hamiltonian vector field can be associated to a family $h_1,\ldots,h_k$ of Hamiltonian functions (each one relative to a different presymplectic form of a $k$--symplectic structure).
It is convenient for the study $k$--symplectic Lie systems to introduce some generalisation of the Hamiltonian function notion for presymplectic forms to deal simultaneously with all $h_1,\ldots,h_k$. In this section, we propose and analyse the properties of such a generalisation. Some of our results extend to our $k$--symplectic structures several theorems devised by Awane in \cite{Aw-1992} for a more particular type of $k$--symplectic structures.

\begin{definition} Given a polysymplectic structure $\Omega=\sum_{i=1}^k\omega_i\otimes e^i$ on $N$, we say that $h=h_1\otimes e^1+\ldots+h_k\otimes e^k$ is an $\Omega$--{\it Hamiltonian function} if there exists a vector field $X_h$ on $N$ such that $\iota_{X_h}\omega_i=dh_i$ for $i=1,\ldots,k$. In this case, we call $h$ an $\Omega$--{\it Hamiltonian function} for $X_h$. We write $C^\infty(\Omega)$ for the space of $\Omega$--Hamiltonian functions.
\end{definition}

We already illustrated that a polysymplectic form $\Omega$ depends on the chosen bases $\{e^1, \ldots, e^k\}$ therefore, also the $\Omega$--Hamiltonian function $h$. Nevertheless, if $\Omega$ and $\tilde{\Omega}$ are two gauge equivalent polysymplectic forms then the sets $C^\infty(\Omega)$ and $C^\infty(\tilde{\Omega})$ are the same up to a change of variables on $\mathbb{R}^k$.

Observe that an $\Omega$--Hamiltonian function is a certain type of $\mathbb{R}^k$-valued Hamiltonian function. In \cite{merino}, the author called $k$--Hamiltonian system associated to the $\mathbb{R}^k$-valued Hamiltonian $h$  the vector field $X_h$ of the above definition. Moreover, Awane \cite{Aw-1992} called $h$ a Hamiltonian map of $X$ when $X$ is additionally an infinitesimal automorphism of a certain distribution on which it is assumed that the presymplectic forms of the $k$--symplectic distribution vanish.

\begin{example}\label{ex1} In view of the relations (\ref{eq1}) and (\ref{eq2}), the vector fields $X_1=4u^2\partial/\partial u+4uv\partial/\partial v+v^2\partial/\partial w$, $X_2=\partial/\partial u$ and $X_3=2u\partial/\partial u+v\partial/\partial v$ have $\Omega$--Hamiltonian functions
$$
\begin{array}{c}
f= \Big(4uw-8\displaystyle\frac{u^2w^2}{v^2}-\displaystyle\frac{v^2}{2}\Big)\otimes e^1 + \Big(4u-16\displaystyle\frac{u^2w}{v^2}\Big) \otimes e^2,\\\noalign{\medskip}
g=-2\displaystyle\frac{w^2}{v^2}\otimes e^1-4\frac{w}{v^2}\otimes e^2,\qquad h=\Big(w-4\frac{uw^2}{v^2}\Big)\otimes e^1- 8\frac{uw}{v^2}\otimes e^2,
\end{array}
$$
relative to the polysymplectic structure $\Omega=\omega_{RS-1}\otimes e^1+\omega_{RS-2}\otimes e^2$ obtained from the two--symplectic structure $(\omega_{RS-1},\omega_{RS-2})$ constructed from the presymplectic forms (\ref{RS-1Sym}).
\end{example}

\begin{proposition}\label{unique} Let $\Omega=\sum_{i=1}^k\omega_i\otimes e^i$ be a polysymplectic structure, every $\Omega$--Hamiltonian vector field is associated, at least, to an
$\Omega$--Hamiltonian function. Conversely, every $\Omega$--Hamiltonian function induces a unique $\Omega$--Hamiltonian vector field.

\end{proposition}
\begin{proof}The direct part is trivial. Let us prove the converse. By definition, each $\Omega$--Hamiltonian function $h=h_1\otimes e^1+\ldots+h_k\otimes e^k$ is associated to, at least, one vector field $X_h$. Suppose that there exist  two $\Omega$--Hamiltonian vector fields $X_1$ and $X_2$ associated to $h$. Then, we have
$$
\iota_{X_1}\omega_i=\iota_{X_2}\omega_i={\rm d}h_i,\qquad i=1,\ldots,k
$$
and
$$
\iota_{X_	1-X_2}\omega_i=0,\qquad i=1,\ldots,k.
$$
Since $\ker \omega_1\cap\ldots\cap \ker \omega_k=\{0\}$, it turns out that $X_1=X_2$.
\end{proof}

\begin{proposition} The space $C^\infty(\Omega)$ is a linear space over $\mathbb{R}$ with the natural operations:
\begin{equation*}
h+g\equiv\sum_{i=1}^k(h_i+g_i)\otimes e^i,\qquad \lambda \cdot h\equiv \sum_{i=1}^k\lambda h_i\otimes e^i
\end{equation*}
where $h=\sum_{i=1}^k h_i\otimes e^i$, $g=\sum_{i=1}^kg_i\otimes e^i\in C^\infty(\Omega)$ and $\lambda\in\mathbb{R}$.
\end{proposition}
\begin{proof}
Let $X_h$ and $X_g$ be the $\Omega$--Hamiltonian vector fields associated to $h$ and $g$, respectively. The linear combination $\lambda h+\mu g$, with $\lambda,\mu\in\mathbb{R}$, is an $\Omega$--Hamilto\-nian function associated to the vector field $\lambda X_h+\mu X_g$. Indeed,
$$
\iota_{\lambda X_h+\mu X_g}\omega_i={\rm d}(\lambda h_i+\mu g_i),\qquad i=1,\ldots,k.
$$
Then, $C^\infty(\Omega)$ is closed with respect to the defined addition of elements and multiplication by scalars. It is immediate that these operations give rise to a vector space structure on $C^\infty(\Omega)$.
\end{proof}

\begin{proposition} The space $C^\infty(\Omega)$ becomes a Lie algebra when endowed with
 the bracket $\{\cdot,\cdot\}_\Omega:C^\infty(\Omega)\times C^\infty(\Omega)\rightarrow C^\infty(\Omega)$ of the form
\begin{equation}\label{LieB}
\{h_1\otimes e^1+\ldots+h_k\otimes e^k,h'_1\otimes e^1+\ldots+h'_k\otimes e^k\}_\Omega=\{h_1,h'_1\}_{\omega_1}\otimes e^1+\ldots+\{h_k,h'_k\}_{\omega_k}\otimes e^k,
\end{equation}
where $\{\cdot,\cdot\}_{\omega_i}$ is the Poisson bracket naturally induced by the presymplectic form $\omega_i$, with $i=1,\ldots,k$.
\end{proposition}
\begin{proof}
Given two $\Omega$--Hamiltonian functions $h=\sum_{i=1}^k h_i\otimes e^i$, $g=\sum_{i=1}^kg_i\otimes e^i$ with $\Omega$--Hamiltonian vector fields $X_h$ and $X_g$, we have
$$
\iota_{[X_h,X_g]}\omega_i={\rm d}\,\{g_i,h_i\}_{\omega_i},\qquad i=1,\ldots,k.
$$
Hence, $\{g, h\}_\Omega$ is an $\Omega$--Hamiltonian function with Hamiltonian vector field $[X_h,X_g]$. So, $C^\infty(\Omega)$ is closed with respect to this bracket, which is trivially antisymmetric and holds the Jacobi  identity, which turns $(C^\infty(\Omega),\{\cdot,\cdot\}_\Omega)$ into a Lie algebra.

\end{proof}

We cannot ensure $C^\infty(\Omega)$ to be a Poisson algebra in a natural way. Observe that given $h=\sum_{i=1}^kh_i\otimes e^i,g=\sum_{i=1}^kg_i\otimes e^i \in C^\infty(\Omega)$, the function
\begin{equation}\label{Prod}
h\cdot g=(h_1g_1)\otimes e^1+\ldots+(h_kg_k)\otimes e^k
\end{equation}
is not in general a $C^\infty(\Omega)$--function. Indeed,
$$
{\rm d}\, (h_ig_i)= g_i{\rm d}\,h_i+h_i{\rm d}\,g_i=\iota_{g_iX_h}\omega_i+\iota_{h_iX_{g}}\omega_i=\iota_{(g_iX_h+h_iX_g)}\omega_i,\qquad i=1,\ldots,k.
$$
In general, $g_iX_h+h_iX_g$ is different for each $i$ and $h\cdot g$ is not an $\Omega$--Hamiltonian function. For instance, consider again Example \ref{ex1}. The
function
$$
h\cdot g=-2\frac{w^2}{v^2} \Big(w-4\frac{uw^2}{v^2}\Big)\otimes e^1+32\frac{uw^2}{v^4}\otimes e^2
$$
is not an $\Omega$--Hamiltonian function for $\Omega=\omega_{RS-1}\otimes e^1+\omega_{RS-2}\otimes e^2$. Indeed, $-2w^2/v^2\left(w-4\frac{uw^2}{v^2}\right)$ and $32uw^2/v^4$ are related to  the vector fields
$$
g_1X_h+h_1X_g\!=\!-\frac{2w^2}{v^2}\Big(2u\frac{\partial}{\partial u}+v\frac{\partial}{\partial v}\Big)+\Big(w-4\frac{uw^2}{v^2}\Big)\frac {\partial}{\partial u}\!,\;\; g_2X_h+h_2X_g\!= \!-\frac{4w}{v^2}\Big(2u\frac{\partial}{\partial u}+v\frac{\partial}{\partial v}\Big)- \frac{8uw}{v^2}\frac {\partial}{\partial u}\!,
$$
which are different.

Since we cannot ensure that $(C^\infty(\Omega),\cdot,\{\cdot,\cdot\}_\Omega)$ is a Poisson algebra, we cannot neither say that $\{h,\cdot\}_{\Omega}:g\in C^\infty(\Omega)\mapsto \{g,h\}_\Omega\in C^\infty(\Omega)$, with $h\in C^\infty(\Omega)$, is a derivation with respect to the product (\ref{Prod}) of $\Omega$--Hamiltonian functions. This shows that $k$--symplectic geometry becomes quite different from Poisson and presymplectic geometry, where an equivalent of this result holds. Nevertheless, we can still ensure that $\{h,g\}_\Omega=0$ for every locally constant function $g$ and, moreover, we can still prove other properties of this Lie algebra. For instace, let us consider the following result.

\begin{proposition} Consider a polysymplectic manifold $(N,\Omega)$. Every $\Omega$--Hamiltonian vector field $X$ acts as a derivation on the Lie algebra $(C^\infty(\Omega),$ $\{\cdot,\cdot\}_\Omega)$ in  the form
$$
Xf=\{f,h\}_\Omega,\qquad \forall f\in C^\infty(\Omega),
$$
with $h$ being an $\Omega$--Hamiltonian function for $X$.
\end{proposition}
\begin{proof} Note that $\{f,h\}_\Omega$ does not depend on the chosen $\Omega$--Hamiltonian for $X$. Every two $\Omega$--Hamiltonian functions related to the same $\Omega$--Hamiltonian vector field differ on a constant (on each connected component on $N$). So, if $h_1$ and $h_2$ are $\Omega$--Hamiltonian functions for $X$, then $\{f,h_1\}_{\Omega}=\{f,h_2\}_{\Omega}$ and $Xf$ becomes well-defined.

Now,
$$
X\{f,g\}_\Omega=\{\{f,g\}_\Omega,h\}_\Omega=\{\{f,h\}_\Omega,g\}_\Omega+\{f,\{g,h\}_\Omega\}_\Omega=\{Xf,g\}_\Omega+\{f,Xg\}_\Omega.
$$
Since $X$ acts linearly on $C^\infty(\Omega)$, the results follows.
\end{proof}

\begin{theorem} Given a polysymplectic form $\Omega=\sum_{i=1}^k\omega_i\otimes e^i$ on a manifold $N$, we can define an exact sequence of Lie algebras:
\begin{equation}\label{ses}
0\hookrightarrow \stackrel{k}{\overbrace{{\rm H}_{\rm dH}^0(N)\oplus{\ldots}\oplus {\rm H}_{\rm dH}^0(N)}}\hookrightarrow C^\infty(\Omega)\stackrel{{B_\Omega}}{\longrightarrow} {\rm Ham}(\Omega)\rightarrow 0,
\end{equation}
where $B_\Omega(f)=-X_f$ is the $\Omega$--Hamiltonian vector field corresponding to $f$ and ${\rm H}_{\rm dH}^0(N)$ is the first De Rham cohomology group of $N$.
\end{theorem}
\begin{proof}
First, we prove that ${\rm Ham}(\Omega)$ is a Lie algebra. In fact, given two $\Omega$--Hamiltonian vector fields $X$ and $Y$ there is two $\Omega$--Hamiltonian functions $h$ and $g$ such that $X=X_h$ and $Y=Y_g$ (see Proposition \ref{unique}). Moreover, we have that $\iota_{X}\omega_i={\rm d}h_i, \iota_Y\omega_i={\rm d}g_i$ for each $i=1,\ldots, k$. Therefore
$$
\begin{aligned}
& \iota_{\lambda X+\mu Y}\omega_i={\rm d}(\lambda h_i+\mu g_i),\\
& \iota_{[X,Y]}\omega_i={\rm d}\{g_i,h_i\}_{\omega_i},\\
\end{aligned} \qquad i=1,\ldots,k, \qquad \forall \lambda, \mu\in\mathbb{R}.
$$
Hence, the sum, the Lie bracket and the multiplication by scalars of $\Omega$--Hamiltonian vector fields are $\Omega$--Hamiltonian vector fields, that is, Ham$(\Omega)$ is a vector space.

Once we have proved that all the spaces in (\ref{ses}) are Lie algebras, we turn to showing that the sequence is exact. The inclusions of $0$ in ${\rm H}_{\rm dH}^0(N)\oplus\stackrel{k}{\ldots}\oplus {\rm H}_{\rm dH}^0(N)$ and of ${\rm H}_{\rm dH}^0(N)\oplus\stackrel{k}{\ldots}\oplus {\rm H}_{\rm dH}^0(N)$ in $C^\infty(\Omega)$ are obviously Lie algebra morphisms. Likewise, the projection of ${\rm Ham}(\Omega)$ onto $0$ is also. If we take into account that the Lie bracket $\{f,g\}_{\Omega}$, where $f$ and $g$ are $\Omega$--Hamiltonian functions with $\Omega$--Hamiltonian vector fields $X_f$ and $X_g$, admits an $\Omega$--Hamiltonian vector field $-[X_f,X_g]$, we obtain that $B_\Omega(\{f,g\}_\Omega)=[X_f,X_g]$. In other words, $B_\Omega$ is a Lie algebra morphism.

Finally, observe that the kernel of $B_\Omega$ is given by those $\Omega$--Hamiltonian functions $h$ related to a zero vector field. That means, that ${\rm d}h_i=0$ for $i=1,\ldots,k$. So, every $h_i$ is constant on each connected component $O_j$, with $j=1,\ldots,p$, of $N$ and its value is determined by a constant on each $O_j$. This gives the isomorphism
$$
\ker B_\Omega\ni h\longmapsto (h_1(O_1),\ldots,h_1(O_p),\ldots,h_k(O_1),\ldots,h_k(O_p))\in {\rm H}_{\rm dH}^0(N)\oplus\stackrel{k}{\ldots}\oplus {\rm H}_{\rm dH}^0(N).
$$
Using this, we clearly see that the given sequence is exact.
\end{proof}

\section{Derived Poisson algebras}\label{DPA}
Given a $k$--symplectic manifold $(N,\omega_1,\ldots,\omega_k)$, we can construct several Poisson algebras on certain subsets of $C^\infty(N)$, the hereafter called derived Poisson algebras. This will become very important in following sections, where such derived Poisson algebras are employed to study the geometric properties of $k$--symplectic Lie systems.

Given a polysymplectic form $\Omega=\sum_{i=1}^k\omega_i\otimes e^i$, where $e^1,\ldots,e^k$ form a basis for $\mathbb{R}^k$, induced by the $k$--symplectic structure  $(\omega_1,\ldots,\omega_k)$, and an element $\theta\in (\mathbb{R}^{k})^*$, it is immediate that the contraction
$$
\Omega_\theta\equiv \langle \Omega,\theta\rangle=\sum_{i=1}^k \theta(e^i)\omega_i
$$ is a presymplectic form on $N$. We consider ${\rm Adm}(\Omega_\theta)$, the set of admissible functions with respect to $(N,\Omega_\theta)$.  We hereafter denote by $X_f$, with $f$ being a function on $N$, a Hamiltonian vector field of $f$ relative to a presymplectic form. Recall that when $f$ is a $k$--Hamiltonian function, $X_f$ denotes the $k$--Hamiltonian vector field associated to $f$.

Note that a vector field $X$ is $k$--Hamiltonian if and only if it is Hamiltonian for all the presymplectic forms of the space $\langle \omega_1,\ldots, \omega_k\rangle$. In particular, $X$ is Hamiltonian for any presymplectic form $\Omega_\theta$ with $\theta\in (\mathbb{R}^k)^*$. This gives rise to the following proposition.
\begin{proposition}
   Let $\Omega = \sum_{i=1}^k\omega_i\otimes e^i$ be a polysymplectic structure and $\theta\in (\mathbb{R}^k)^*$. Every $\Omega$--Hamiltonian function gives rise to an admissible function with respect to $(N,\Omega_\theta)$.
\end{proposition}
\begin{proof}
    If $h=h_1\otimes e^1+ \ldots + h_k\otimes e^k$ is an $\Omega$--Hamiltonian function, then there exists an $\Omega$--Hamiltonian vector field $X_h$ such that
    \[
        \iota_{X_h}\omega_i={\rm d}h_i,\qquad i=1,\ldots, k\,.
    \]

    Thus, one has
    \[
        \iota_{X_h}\Omega_\theta=\displaystyle\sum_{i=1}^k\theta(e^i)\iota_{X_h}\omega_i= \displaystyle\sum_{i=1}^k\theta(e^i){\rm d}h_i ={\rm d} h_\theta,
    \]
    where $h_\theta=\langle h,\theta \rangle=\displaystyle\sum_{i=1}^k\theta(e^i)h_i$. Therefore $h_\theta\in {\rm Adm}(\Omega_\theta)$.
\end{proof}
\begin{proposition} Let $(\omega_1,\ldots,\omega_k)$ be a $k$--symplectic structure and $\{e^1,\ldots,e^k\}$ a basis of $\mathbb{R}^k$.  The $k$--symplectic structure induces a family of Poisson algebras $({\rm Adm}(\Omega_\theta),\cdot,\{\cdot,\cdot\}_\theta)$, where $\{\cdot,\cdot\}_\theta$ is the Poisson bracket induced by the presymplectic form $\Omega_\theta$ on its space of admissible functions.
\end{proposition}
\begin{proof}
    It is immediate that the sum and multiplication by scalars of admissible functions are admissible functions. This turns ${\rm Adm}(\Omega_\theta)$ into a vector space.  The product of functions is bilinear and commutative. Moreover, if $h$ and $g$ are admissible functions with Hamiltonian vector fields $X_h$ and $X_g$, then $h\cdot g$ is an admissible function with Hamiltonian vector field $gX_h+hX_g$. So ${\rm Adm}(\Omega_\theta)$ is an $\mathbb{R}$-algebra.

    If $h$ and $g$ are admissible functions with respect to $(N,\Omega_\theta)$ with Hamiltonian vector fields $X_h$ and $X_g$, then $\{h,g\}_\theta$ is an admissible function with Hamiltonian vector field $[X_g,X_h]$. So, for each $\theta\in (\mathbb{R}^k)^*$, ${\rm Adm}(\Omega_\theta)$ is closed with respect to this bracket, which is antisymmetric and holds the Jacobi and the Leibniz identity. Then, $({\rm Adm}(\Omega_\theta),\cdot,\{\cdot,\cdot\}_\theta)$ is a Poisson algebra for each $\theta$.
\end{proof}

\begin{proposition}\label{poissonderviation} Given a polysymplectic form $\Omega=\sum_{i=1}^k\omega_i\otimes e^i$, every $\Omega$--Hamiltonian vector field $X_h$ is a derivation on all the Lie algebras $({\rm Adm}(\Omega_\theta),\{\cdot,\cdot\}_\theta)$ with $\theta\in (\mathbb{R}^k)^*$ in  the form
\begin{equation}\label{poissontheta}
X_h f=\{f,h_\theta\}_\theta,\qquad \forall f\in {\rm Adm}({\Omega_\theta}).
\end{equation}
\end{proposition}
\begin{proof} Note that $X_hf=\{f,h_\theta\}_\theta\in {\rm Adm}({\Omega_\theta})$  is well defined. Every two functions related to the same Hamiltonian vector field with respect to $\Omega_\theta$ differ in a locally constant function. So, if $X_{h^1}=X_{h^2}$, then $X_{h^1}f=\{f,h^1_{\theta}\}_{\theta}=\{f,h^2_{\theta}\}_{\theta}=X_{h^2}f$. Now,
$$
X_h\{f,g\}_{\theta}=\{\{f,g\}_{\theta},h_{\theta}\}_{\theta}=
\{\{f,h_{\theta}\}_{\theta},g\}_{\theta}+
\{f,\{g,h_{\theta}\}_{\theta}\}_{\theta}=
\{X_hf,g\}_{\theta}+\{f,X_hg\}_{\theta}.
$$
\end{proof}

\begin{proposition}\label{algebra} Given a polysymplectic form $\Omega=\sum_{i=1}^k\omega_i\otimes e^i$, then
$$
\begin{array}{rccc}
\phi_\theta&:(C^\infty(\Omega),\{\cdot,\cdot\}_\Omega)&\rightarrow &({\rm Adm}({\Omega_\theta}),\{\cdot,\cdot\}_{\theta})\\
&h&\mapsto &h_\theta=\langle h,\theta \rangle\\
\end{array}
$$
is a Lie algebra morphism. Hence, every finite-dimensional Lie algebra $(\mathcal{W}\subset C^\infty(\Omega),\{\cdot,\cdot\}_\Omega)$  is a Lie algebra extension of the Lie algebra $(\phi_\theta(\mathcal{W}),\{\cdot,\cdot\}_{\theta})$.
\end{proposition}
\begin{proof}
    Let $g,h$ be two $\Omega$--Hamiltonian functions. From (\ref{LieB}) and (\ref{poissontheta}), we obtain
    \[
        \phi_\theta\big(\{h,g\}_\Omega\big)= \sum_{i=1}^k\theta(e^i)\{h_i,g_i\}=\sum_{i=1}^k\theta(e^i)X_g(h_i)=X_g(h_\theta)=\{h_\theta,g_\theta\}_\theta\,,
    \]
    and $\phi_\theta$ becomes a Lie algebra morphism. Moreover, we have the exact sequence of Lie algebras
    $$
    0\hookrightarrow \stackrel{k}{\overbrace{\left({\rm H}_{\rm dR}^0(N)\oplus{\ldots}\oplus {\rm H}_{\rm dR}^0(N)\right)}}\cap \mathcal{W} \hookrightarrow \mathcal{W}\stackrel{\phi_\theta|_\mathcal{W}}{\longrightarrow} \phi_\theta(\mathcal{W})\rightarrow 0.
    $$
    Therefore, $(\mathcal{W},\{\cdot,\cdot\}_\Omega)$ is a Lie algebra extension of $(\phi_\theta(\mathcal{W}),\{\cdot,\cdot\}_{\theta})$.
\end{proof}
\newpage
\begin{proposition} For every polysymplectic manifold $(N,\Omega)$ we have the following commutative exact diagram:

{\tiny
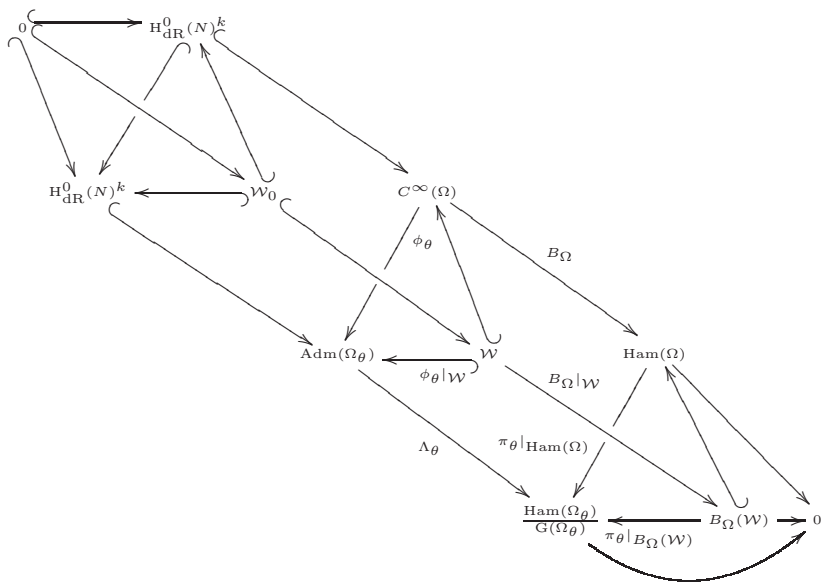
$$
\xymatrix@C=0.08cm@R=1.7cm{
0 \ar@{^(->}[rrrd] \ar@<1ex>@{^(->}[rr]\ar@<-1ex>@{_(->}[rd] &    & {\rm H_{dR}^0}(N)^k \ar@<1ex>@{^(->}[rrrd] \ar@{^{(}->}[dl]|!{[ll];[rd]}\hole&   &   &  &   &   & & &&\\
& {\rm H_{dR}^0}(N)^k   \ar@<-1ex>@{^(->}[rrrd] &   & \mathcal{W}_0 \ar@<-0.8ex>@{_{(}->}[lu] \ar@{^{(}->}[ll] \ar@<-0.5ex>@{^{(}->}[rrrd]&   & C^\infty(\Omega)\ar[dl]^-(0.2){\phi_\theta}|!{[ll];[rd]}\hole \ar@<1ex>[rrrd]^-{B_\Omega}&   &   & & &&\\
&    &   &   & {\rm Adm}(\Omega_\theta) \ar[rrrd]_-{\Lambda_\theta} &   & \mathcal{W} \ar@<-0.8ex>@{_{(}->}[lu] \ar@<1ex>@{^{(}->}[ll]^(0.3){\phi_\theta|_\mathcal{W}} \ar[rrrd]^(.25){B_\Omega|_\mathcal{W}}&   &{\rm Ham}(\Omega)\ar[dl]_(0.6){\pi_\theta|_{{\rm Ham}(\Omega)}}|!{[ll];[rd]}\hole \ar@<1ex>[rrrd]& &&\\
&    &   &   &   &   &   &  \frac{{\rm Ham}(\Omega_\theta)}{{\rm G}(\Omega_\theta)}\ar@/_{8mm}/[rrrr] & &B_\Omega(\mathcal{W}) \ar@<-0.8ex>@{_{(}->}[lu] \ar[ll]^{\pi_\theta|_{B_\Omega(\mathcal{W})}} \ar[rr]&& 0\\
}
$$}
\bigskip

%\xymatrix@R=0.75cm@C=0.30cm{
% &&H^ 0(N)^p\ar[ldd]\ar@{^{(}->}[rrd]&&&&&&&&\\
%0\ar@{^{(}->}[urr]\ar@{_{(}->}[dr]\ar@{_{(}->}[rrr]&&&H^ 0(N)^p\cap \mathcal{W}\ar@{_{(}->}[lu]\ar@{_{(}->}[rrrd]\ar@{^{(}->}[lld]&C^\infty(\Omega)\ar[ddl]^{\phi_\theta}\ar[drrr]^{B_\Omega}&&&&&&\\
%&H^ 0(N)^p\ar@{_{(}->}[rrd]&&&&&\mathcal{W}\ar@{_{(}->}[llu]\ar@{^{(}->}[dlll]^{\phi_\theta|_\mathcal{W}}\ar[rrrd]_{B_\Omega|_\mathcal{W}}&{\rm Ham}(\Omega)\ar[ddl]^{\pi_\theta|_{{\rm Ham}(\Omega)}}\ar[ddrrr]&&&&\\
%&&&{\rm Adm}(\Omega_\theta)\ar[rrrd]^\Lambda&&&&&&B_\Omega(\mathcal{W})\ar@{_{(}->}[ull]\ar[rd]\ar[dlll]^{\pi_\theta|_{B_\Omega(\mathcal{W})}}&&\\
%&&&&&&\frac{{\rm Ham}(\Omega_\theta)}{{\rm G}(\Omega_\theta)}\ar[rrrr]&&&&0}
\noindent where $\mathcal{W}_0={\rm H}^0(N)^k\cap \mathcal{W}$ and we recall that ${\rm G}(\Omega_\theta)$ is the space of gauge vector fields of $\Omega_\theta$, we call $\pi_\theta:X\in {\rm Ham}(\Omega_\theta)\mapsto [X]\in {\rm Ham}(\Omega_\theta)/G(\Omega_\theta)$  the quotient map onto ${\rm Ham}(\Omega_\theta)/G(\Omega_\theta)$, and $\Lambda_\theta:{\rm Adm}(\Omega_\theta)\rightarrow{\rm Ham}(\Omega_\theta)/G(\Omega_\theta)$ is the
Lie algebra morphism mapping each $f\in {\rm Adm}(\Omega_\theta)$ to the class $[-X_f]$.
\end{proposition}
\begin{proof}
The only non-trivial part which does not follow from previous results of this section is to prove that the diagram
$$
\xymatrix{
C^\infty(\Omega) \ar[r]^-{B_\Omega} \ar[d]^-{\phi_\theta}& {\rm Ham}(\Omega)\ar[d]^-{\pi_\theta|_{{\rm Ham}(\Omega)}}\\
{\rm Adm}(\Omega_\theta)\ar[r]^-{\Lambda_\theta} & \frac{{\rm Ham}(\Omega_\theta)}{G(\Omega_\theta)}\\
}
$$
is commutative. Using that ${\rm Ham}(\Omega)\subset {\rm Ham}(\Omega_\theta)$, we have that
$$
{\rm d}h_\theta ={\rm d}\phi_\theta(h)=\langle {\rm d}h,\theta\rangle= \langle \iota_{X_h}\Omega,\theta\rangle=\iota_{X_h}\Omega_\theta\Rightarrow [X_{h_\theta}]=[X_h],
$$
for an arbitrary $h\in C^\infty(\Omega)$. So,
$$
\pi_\theta\circ B_\Omega(h)=\pi_\theta(-X_h)=[-X_h]=[-X_{h_\theta}]=\Lambda_\theta(h_\theta)=\Lambda_\theta\circ \phi_\theta(h)
$$
and $\pi_\theta\circ B_\Omega=\Lambda_\theta\circ \phi_\theta$.
\end{proof}

\section{k-Hamiltonian Lie structures}\label{HLS}

Let us further investigate the properties of the $k$--symplectic Lie systems provided in the
previous sections. Consider again the Schwarzian equations in first-order form (\ref{firstKS3}). Remind that $Y_1, Y_2$ and $Y_3$ are Hamiltonian vector fields with respect to the presymplectic structures $\omega_1$ and $\omega_2$. In particular, from the relations (\ref{3KSHamFun}) and (\ref{3KSHamFunOmega2}), the vector fields $Y_1,\, Y_2$ and $Y_3$ have Hamiltonian functions
\begin{equation}\label{3KSHamFunloc}
h^1_{1}=\frac{2}{v},\quad h^2_{1}=\frac{a}{v^2},\quad h^3_{1}=\frac{a^2}{2v^3},
\end{equation}
and
\begin{equation}\label{3KSHamFunloc2}
h^1_{2}=-\frac{4x}{v},\quad h^2_{2}=2-\frac{2ax}{v^2},\quad h^3_{2}=\left(\frac{2a}{v}-\frac{a^2x}{v^3}\right),
\end{equation}
with respect to the presymplectic forms $\omega_1$ and $\omega_2$ given by (\ref{forms}), correspondingly. Moreover, we have
\begin{equation*}
 \begin{aligned}
\left\{h_{i}^1,h_{i}^2\right\}&=-h_{i}^1,\qquad \left\{h_{i}^1,h_{i}^3\right\}&=-2h_{i}^2,\qquad \left\{h_{i}^2,h_{i}^3\right\}&=-h_{i}^3,\qquad i=1,2.
\end{aligned}
\end{equation*}
Consequently, the functions $h^\alpha_{i}$, with $\alpha=1,2,3$ and a fixed $i$, span a finite-dimensional real Lie algebra of functions isomorphic to $\mathfrak{sl}(2,\mathbb{R})$. The same applies for $h^\alpha_{1}+h^\alpha_{2}$, with $\alpha=1,2,3$, and in general for any linear combination $\mu_1 h_{1}^\alpha+\mu_2h_{2}^\alpha$, with fixed $(\mu_1,\mu_2)\in \mathbb{R}^2\backslash \{(0,0)\}$. In this way, a $k$--symplectic structure is associated with many different Lie algebras of functions which can be employed to study the properties of the system.

Now, we consider the space $C^\infty(\Omega)$ of $\Omega$--Hamiltonian functions given by the two-symplectic structure $(\omega_1,\omega_2)$. From the relations (\ref{3KSHamFun}) and (\ref{3KSHamFunOmega2}), the functions
\[
h^\alpha=h^\alpha_{1}\otimes e^1 + h^\alpha_{2}\otimes e^2,
\]
with $\alpha=1,2,3$, span a finite-dimensional Lie algebra when endowed with the Lie bracket (\ref{LieB}).

Thus, every $X_t^{3KS}$ is an $\Omega$--Hamiltonian vector field with $\Omega$--Hamiltonian function
\[
    h_t^{3KS}= (h^3_{1}+b_1(t)h^1_{1})\otimes e^1 + (h^3_{2}+b_1(t)h^1_{2})\otimes e^2.
\]
Since we assume $b_1(t)$ to be non-constant, the space Lie$(\{h_t^{3KS}\}_{t\in \mathbb{R}},\{\cdot,\cdot\}_\Omega)$ becomes a real Lie algebra isomorphic to $\mathfrak{sl}(2,\mathbb{R})$.

If we now turn to the system of Riccati equations (\ref{3Ricc}), we see that we can obtain a similar result. More specifically, the relations (\ref{RHamFun}) and (\ref{RHamFunOmega2}) imply that $X_1,\, X_2$ and $X_3$ have Hamiltonian functions
\begin{equation}\label{3RHamFunloc}
h_1^{1}=\frac{1}{x_1-x_2}+\frac{1}{x_3-x_4},\quad h^2_{1}=\frac 12\left(\frac{x_1+x_2}{x_1-x_2}+\frac{x_3+x_4}{x_3-x_4}\right),\quad h^3_{1}=\frac{x_1 x_2}{x_1-x_2}+\frac{x_3 x_4}{x_3-x_4}
\end{equation}
and
\begin{equation}\label{3RHamFunloc2}
h^1_{2}= \sum_{\stackrel{i,j=1}{i\leq j}}^4\frac{1}{x_i-x_j},\quad h^2_{2} = \frac 12\left(\sum_{\stackrel{i,j=1}{i\leq j}}^4\frac{x_i+x_j}{x_i-x_j}\right),\quad h^3_{2}= \sum_{\stackrel{i,j=1}{i\leq j}}^4\frac{x_i x_j}{x_i-x_j}\,.
\end{equation}
Moreover, we have that
\begin{equation*}
 \begin{aligned}
\{h_{i}^1,h_{i}^2\}&=-h_{i}^1,\qquad \{h_{i}^1,h_{i}^3\}&=-2h_{i}^2,\qquad \{h_{i}^2,h_{i}^3\}&=-h_{i}^3,\qquad i=1,2.
\end{aligned}
\end{equation*}
Consequently, the functions $h^\alpha_{i}$, with $\alpha=1,2,3$ and a fixed $i$ span a finite-dimensional real Lie algebra of functions.

Now, we consider the space $C^\infty(\Omega)$ of $\Omega$--Hamiltonian functions given by the two-symplectic structure $(\omega_1,\omega_2)$. From the relations (\ref{RHamFun}) and (\ref{RHamFunOmega2}), the functions
\[
h^\alpha=h^\alpha_{1}\otimes e^1 + h^\alpha_{2}\otimes e^2,
\]
with $\alpha=1,2,3$ span a finite-dimensional Lie algebra when endowed with the Lie bracket (\ref{LieB}).

Thus, every $X_t^{R}$ is an $\Omega$--Hamiltonian vector field with $\Omega$--Hamiltonian function
\[
    h_t^{R}= a(t)h^1+b(t)h^2+c(t)h^3\,.
\]
Again, we can associate $X$ to a curve $t\to h_t^{R}$ in a finite-dimensional  real Lie algebra  Lie$(\{h_t^{R}\}_{t\in \mathbb{R}},\{\cdot,\cdot\}_\Omega)$.

Both examples suggest us to define the following notions.
\begin{definition}
A \textit{$k$--symplectic Lie--Hamiltonian structure} is a triple $(N,\Omega, h)$ where $(N,\Omega)$ is a polysymplectic manifold and $h$ represents a $t$-parametrized family of $\Omega$--Hamiltonian functions $h_t\colon N\to \mathbb{R}^k$ such that Lie$(\{h_t\}_{t\in \mathbb{R}}, \{\cdot, \cdot\}_\Omega)$ is a finite-dimensional real Lie algebra.
\end{definition}
\begin{definition}
    A $t$-dependent vector field $X$ is said to admit a $k$--symplectic Lie--Hamiltonian structure $(N,\Omega,h)$ if $B_\Omega(h_t)=-X_t$, for all $t\in \mathbb{R}$.
\end{definition}
\begin{theorem}\label{HamLieSys}$\!$ A system $X\!$  admits a $k$--symplectic Lie--Hamiltonian structure if and only if it is a $k$--symplectic Lie system.
\end{theorem}

\begin{proof}

Let $(N, \Omega, h)$ be a $k$--symplectic Lie--Hamiltonian structure for $X$, then ${\rm Lie}(\{h_t\}_{t\in \mathbb{R}}, \{\cdot, \cdot\}_\Omega)$ is a finite-dimensional real Lie algebra. Since $B_\Omega$ is  a Lie algebra morphism, then $V= B_\Omega ({\rm Lie}(\{h_t\}_{t\in \mathbb{R}}))$ is  a finite-dimensional real Lie algebra. As each vector field $-X_t$ is $\Omega$--Hamiltonian with an $\Omega$--Hamiltonian function within $\{ h_t\}_{t\in\mathbb{R}}$, then  $\{X_t\}_{t\in\mathbb{R}}\subset V$. Therefore, $V^X={\rm  Lie}(\{X_t\}_{t\in \mathbb{R}})\subset V$ and $X$ is a $k$--symplectic Lie system.

Conversely, if $X$ is a $k$--symplectic Lie system, the vector fields $\{X_t\}_{t\in\mathbb{R}}$ are contained in a finite-dimensional real Lie algebra
 of $\Omega$--Hamiltonian vector fields $V^X$. So, we can write  $X_t=\sum_{\alpha=1}^rb_\alpha(t)X_\alpha$ for a basis $\{X_1,\ldots,X_r\}$ for $V^X$ of $\Omega$--Hamiltonian vector fields and certain $t$-dependent functions $b_1,\ldots,b_r$. In view of the sequence (\ref{ses}), $B_\Omega^{-1}(V^X)$ is a finite-dimensional real Lie algebra of $\Omega$--Hamiltonian
 functions. If $h_1,\ldots,h_r$ is a set of elements of $C^\infty(\Omega)$ with associated $\Omega$--Hamiltonian vector fields $X_1,\ldots,X_r$, then $h_t=\sum_{\alpha=1}^rb_\alpha(t)h_\alpha$ is an $\Omega$--Hamiltonian function and $-B_\Omega(h_t)=X_t$ for every $t\in\mathbb{R}$. Hence, $h_1,\ldots,h_r$ are contained within the finite-dimensional Lie algebra $B_\Omega^{-1}(V^X)$ and $(N,\Omega,h)$ becomes a $k$--symplectic Lie--Hamiltonian structure for $X$.
\end{proof}

\section{On general properties of $k$--symplectic Lie systems}\label{KSLS}

We now turn to describing the analogue for $k$--symplectic Lie systems of the basic
properties of general Lie systems. Additionally, we show how the derived algebras enable us to investigate their $t$-independent constants of motion.

Recall that, as for every Lie system, the general solution $x(t)$ of a $k$--symplectic Lie system $X$ on $N$ can be brought into the form $x(t)=\varphi(g(t),x_0)$, where $x_0\in N$ and  $\varphi\colon  G\times N \to N$ is a Lie group action. The $\varphi$ plays another relevant r\^ole. It is known that if
$G$ is connected, every curve $\bar g(t)$ in $G$ induces a
 $t$-dependent change of variables mapping a Lie system $X$ taking values in a
Lie algebra $V^X$ into another Lie system $Y$,
 with general solution $ y(t)=\varphi(\bar g(t),x(t))$, taking values in the
same Lie algebra $V^X$ \cite{CGL09,CarRamGra,CLLEmd}.
 In the particular case of $X$ being a $k$--symplectic Lie system, we have that $V^X$ consists of $k$--Hamiltonian vector fields with respect to some $k$--symplectic structure. Since the vector fields
$\{Y_t\}_{t\in\mathbb{R}}$ belong to $V^X$ also, they are $k$--Hamiltonian vector fields and $Y$ is  again a $k$--symplectic Lie system.

Using again that $x(t)=\varphi(g(t),x_0)$, we see that the each particular solution of a
Lie system $X$ is contained within an orbit $S$ of $\varphi$. Indeed, it is easy to
see that the vector fields $\{X_t\}_{t\in\mathbb{R}}$ are tangent to such orbits and it makes sense to define the restriction $X|_S$ of $X$ to each orbit $S$.
Therefore, the integration of a Lie system $X$ reduces to integrating its
restrictions to each orbit of $\varphi$, which are Lie systems also. So, it is interesting to know
whether $X|_S$ is again a $k$--symplectic Lie system. More generally, we want to know whether the restriction $X|_S$ of a $k$--symplectic Lie system $X$ to a submanifold $S\subset N$, where it has sense to define $X|_S$, is again a $k$--symplectic Lie system. This requires studying the notion of $l$--symplectic submanifold  ($l\leq k$) of a $k$--symplectic manifold $(N,\omega_1,\ldots, \omega_k)$.

\begin{definition}
    Given a $k$--symplectic manifold $(N,\omega_1,\ldots,\omega_k)$, a submanifold $S\subset N$ is said to be an	 $l$--symplectic submanifold with respect to $(N,\omega_1,\ldots,\omega_k)$, ($l\leq k$) if $\dim S=n_l(l+1)$ for an integer $n_l$ and
    \begin{equation}\label{l_sympl}
        (T_pS)^{\perp,l}\cap T_pS=\{0\},\qquad \forall p\in S,
    \end{equation}
    where $(T_pS)^{\perp,l}$ is the $l$--th orthogonal complement of $T_pS$ with respect to the $k$--symplectic structure $(N,\omega_1,\ldots,\omega_k)$, i.e. $T_pS^{\perp, l}=\{v\in T_pN: \omega_1(v,w)=\ldots=\omega_l(v,w)=0, \forall w\in T_pS\}$ \cite{LV-2012}.
\end{definition}

Let us observe that the condition (\ref{l_sympl}) is equivalent to
    \begin{equation}\label{Con}
\bigcap_{i=1}^l(T_pS)^{\perp_i}\cap T_pS=\{0\},\qquad \forall p\in S,
\end{equation}
where $(T_pS)^{\perp_i}$ is the presymplectic annihilator of $T_pS$, i.e. $T_pS^{\perp_i}=\{v\in T_pN: \omega_i(v,w)=0, \forall w\in T_pS\}$.

It is easy to prove the following
\begin{lemma}
    If $(N,\omega_1,\ldots, \omega_k)$ and $(N,\omega'_1,\ldots, \omega'_k)$ are gauge equivalent and $S\subset N$ is a submanifold then
    $$
        (T_pS)^{\perp,k} = (T_pS)^{\perp',k},\qquad \forall p\in S,
    $$
    where $(T_pS)^{\perp,k}$ and $(T_pS)^{\perp',k}$ are the $k$--th orthogonal $k$--symplectic to $T_pS$ with respect to $(N,\omega_1,\ldots, \omega_k)$ and $(N,\omega'_1,\ldots, \omega'_k)$ respectively.
\end{lemma}

Notice that $
        (T_pS)^{\perp,l} \neq (T_pS)^{\perp',l}$ in general for $l<k$. For instance, consider the linear example given by $N=\mathbb{R}^3$ with the gauge equivalent two-symplectic linear structures
        $(\omega_1=e^1\wedge e^3, \omega_2=e^2\wedge e^3)$ and $(\omega'_1=e^2\wedge e^3, \omega'_2=e^1\wedge e^3)$, where $\{e^1,e^2,e^3\}$ its the dual of the canonical basis of $\mathbb{R}^3$. Then if $S={\rm span}\{e^1\}$, we obtain
        \[
            S^{\perp,1}={\rm span}\{e_1,e_2\}\qquad  \makebox{and}\qquad S^{\perp',1}=\mathbb{R}^3\,.
        \] Therefore, $S^{\perp,1}\neq S^{\perp',1}$.

\begin{lemma}
Given a $k$--symplectic manifold $(N,\omega_1,\ldots,\omega_k)$ and a submanifold $S\subset N$, with $\iota:S\hookrightarrow N$ a
natural embedding, $(\iota^*\omega_1,\ldots,\iota^*\omega_l)$ is an $l$--symplectic structure on $S$ if and only if $S$ is an $l$--symplectic submanifold of $(N,\omega_1,\ldots,\omega_k)$.
\end{lemma}
\begin{proof}
It is a direct consequence of the following relation
\[
    \bigcap_{i=1}^l\ker (\iota^*\omega_i(p))= \bigcap_{i=1}^l\ker (\omega_i(p))\cap T_pS=(T_pS)^{\perp, l}\cap T_pS,\qquad \forall p\in S\,.
\]
\end{proof}
Observe that if a submanifold $S\subset M$ is endowed with an $l$--symplectic structure $(\iota^*\omega_1,\ldots, \iota^*\omega_l)$ with $l< k$, then for all $l'$ such that $l\leq l'\leq k$ (it is necessary that there exists $n_{l'}$ such that ${\rm dim}\, S=n_{l'}(l'+1)$),  $(\iota^*\omega_1,\ldots, \iota^*\omega_{l'})$ is an $l'$--symplectic structure on $S$.

\begin{proposition} Let $(\omega_1,\ldots, \omega_k)$ be a $k$--symplectic structure on $N$ and $X$ be a $k$--symplectic Lie system. Given an $l$--symplectic submanifold $S$ such that $\mathcal{D}^X\subset TS$, the restriction of $X$ to $S$ is an $l$--symplectic Lie system.
\end{proposition}
\begin{proof}
    Let $X$ be a $k$--symplectic Lie system on $N$ with respect to the $k$--symplectic structure $(\omega_1,\ldots, \omega_k)$. Then, $V^X$ is a finite-dimensional real Lie algebra of $k$--Hamiltonian vector fields, i.e. if $Y\in V^X$, then $Y$ is a Hamiltonian vector field with respect to $\omega_1,\ldots, \omega_k$. Since $\mathcal{D}^X\subset TS$, we have that there exists $X|_S$ and $V^{X\vert_S}={\rm Lie}(\{X_t\vert_S\})$ is a finite-dimensional real Lie algebra of Hamiltonian vector fields with respect to $\iota^*\omega_1,\ldots, \iota^*\omega_l$, with $\iota$ being the embedding $\iota:S\hookrightarrow N$. Therefore, the restriction of $X$ to $S$ is an $l$--symplectic Lie system.
\end{proof}

Let us now turn to describing several properties of constants of motion for Lie
systems.

\begin{proposition}\label{IntLie} Let $X$ be a $k$--symplectic Lie system on a manifold $N$ with $k$--symplectic Lie--Hamilto\-nian structure $(N,\Omega, h)$. For each $\theta\in (\mathbb{R}^k)^*$, the space $\mathcal{I}^X_\theta\!$ of $t$-independent
  constants of motion of $X$  admissible relative to $\Omega_\theta$ is a Poisson algebra with respect to each Poisson bracket $\{\cdot,\cdot\}_\theta$ induced by $\Omega_\theta$.
\end{proposition}
\begin{proof}
Let  $f_1,f_2:N\rightarrow\mathbb{R}$ be two $t$-independent  constants of motion
for $X$, i.e. $X_tf_i=0$, for $i=1,2$ and $t\in \mathbb{R}$. As $X$ is a
$k$--symplectic Lie system, all the elements of $V^X$ are Hamiltonian vector fields with respect to each $\Omega_\theta$ with $\theta\in (\mathbb{R}^k)^*$.
Hence, we can write $X_t\{f,g\}_\theta=\{X_tf,g\}_\theta+\{f,X_tg\}_\theta$ for every
$f,g\in {\rm Adm}(\Omega_\theta)$ and $t\in\mathbb{R}$ (see Proposition \ref{poissonderviation}). In particular,
$X_t\{f_1,f_2\}_\theta=\{X_tf_1,f_2\}_\theta+\{f_1,X_tf_2\}_\theta=0$ for every $t\in\mathbb{R}$, i.e. the
Poisson bracket of $t$-independent constants of motion admissible relative to $\Omega_\theta$ is a new one.
As $\lambda f_1+\mu f_2$ and $f_1\cdot f_2$ are also $t$-independent constants
of motion for every $\lambda,\mu\in\mathbb{R}$, it follows that
$\mathcal{I}^X_\theta$ is a Poisson algebra with the bracket $\{\cdot,\cdot\}_\theta$ induced by the presymplectic form $\Omega_\theta$.

\end{proof}

Let us prove some final interesting results about the $t$--independent constants
of
motion for $k$--symplectic Lie systems.
\begin{proposition} Let $X$ be a $k$--symplectic Lie system on a manifold $N$ with $k$--symplectic Lie Hamiltonian structure $(N,\Omega, h)$. For each $\theta\in (\mathbb{R}^k)^*$, the function $f:N\rightarrow
\mathbb{R}$ is a constant of motion for $X$ admissible relative to $\Omega_\theta$ if and only if $f$ Poisson commutes
with all elements of each $\phi_\theta({\rm
Lie}(\{h_t\}_{t\in\mathbb{R}},\{\cdot,\cdot\}_{\Omega}))$.
\end{proposition}
\begin{proof}
The function $f$ is a $t$-independent constant of motion for $X$ if and only if
\begin{equation}\label{con2}
0=X_tf=\{f,\langle h_t,\theta\rangle\}_{\theta},\qquad \forall t\in\mathbb{R},\qquad \forall \theta\in (\mathbb{R}^k)^*.
\end{equation}
From here,
$$
\{f,\{\langle h_t,\theta\rangle ,\langle h_{t'},\theta\rangle\}_{\theta}\}_{\theta}=\{\{f,\langle h_t,\theta\rangle \}_{\theta},\langle h_{t'},\theta\rangle\}_{\theta
}+\{\langle h_t,\theta\rangle,\{f,\langle h_{t'},\theta\rangle \}_{\theta}\}_{\theta}=0,\qquad \forall t,t'\in\mathbb{R},
$$
and inductively follows that $f$ Poisson commutes with all  successive Poisson
brackets of elements of $\{\langle h_t,\theta\rangle\}_{t\in\mathbb{R}}$ and their linear combinations.
As these elements span $\phi_\theta({\rm
Lie}(\{h_t\}_{t\in\mathbb{R}}))$, we get that $f$ Poisson commutes with $\phi_\theta({\rm
Lie}(\{h_t\}_{t\in\mathbb{R}}))$.

Conversely, if $f$ Poisson commutes with $\phi_\theta({\rm Lie}(\{h_t\}_{t\in\mathbb{R}}))$,
it Poisson commutes with the elements $\langle \{h_t\}_{t\in\mathbb{R}},\theta\rangle$, and, in view
of (\ref{con2}), it becomes a constant of motion for $X$ admissible relative to $\Omega_\theta$.
\end{proof}

Observe that every autonomous Hamiltonian system is a $k$--symplectic Lie system with respect a symplectic form $\omega$. Thus, it possesses
a $k$--Hamiltonian structure $(N,\Omega,h)$ with $h$ being a $t$-independent Hamiltonian. In consequence, the above proposition
shows that the time-independent first-integrals for a Hamiltonian system are those functions that Poisson commute with its Hamiltonian, recovering as a particular case this well-known result.

\section{Prolongations of k--symplectic Lie systems}
Our concern in this section is to prove that given a $k$--symplectic Lie system, its prolongations are also $k$--symplectic Lie systems. This enables us to apply our techniques to obtain some of their $t$-independent constants of motion and, through them, the superposition rule for initial (non-prolongated) $k$--symplectic Lie system \cite{CGM07}. To do so, let us define the prolongation of a section of a vector bundle (see \cite{CGLS13} for details).

Let $\tau:E\to N$ be a vector bundle. Its {\it diagonal prolongation} to $N^m$ is the Cartesian product bundle $E^{[m]}=E\times\cdots\times E$ of $m$ copies of $E$, viewed as a vector bundle over $N^m$ in a natural way:
 \begin{equation*}E^{[m]}_{(x_{(1)},\dots,x_{(m)})}\simeq E_{x_{(1)}}\oplus\cdots\oplus E_{x_{(m)}}\,.
\end{equation*}
Every section $e:N\to E$ of $E$ has a natural {\it diagonal prolongation} to a section $e^{[m]}$ of $E^{[m]}$:
\begin{equation*}
e^{[m]}(x_{(1)},\dots,x_{(m)})=e(x_{(1)})+\cdots +e(x_{(m)})\,.
\end{equation*}
Given a function $f:N\rightarrow \mathbb{R}$, we call {\it diagonal prolongation} of $f$ to $N^m$ the function
$\widetilde{f}^{[m]}(x_{(1)},\ldots,x_{(m)})= f(x_{(1)})+\ldots+f(x_{(m)})$.

The above construction can be used to define the diagonal prolongation of a $t$-dependent vector field $X$ on $N$, let us say
$$
X_t=\sum_{l=1}^nX^l(t,x)\frac{\partial}{\partial x^l}.
$$
Its diagonal prolongation to $N^m$ is the unique $t$-dependent vector field $\widetilde X^{[m]}$ on $N^m$ such that $\widetilde X^{[m]}_t=X^{[m]}_t$ for each $t\in\mathbb{R}$, namely
$$
X_t^{[m]}=\sum_{a=1}^m\sum_{l=1}^nX^l(t,x_{(a)})\frac{\partial}{\partial x_{(a)}^l},
$$
where $\{x^l_{(a)}\mid a=1,\ldots,m, l=1,\ldots,n=\dim N\}$ is the coordinate system on $N^m$ given by
defining $x^l_{(a)}(x_{(1)},\ldots,x_{(m)})=x^l(x_{(a)})$ for points $x_{(1)},\ldots,x_{(m)}\in N$.

\begin{proposition}\label{Theo:prol} If $X$ is a $k$--symplectic Lie system relative to $(\omega_1,\ldots,\omega_k)$, then
$\widetilde{X}^{[m]}$ is a $k$--symplectic Lie system relative to $(\omega_1^{[m]},\ldots,\omega_k^{[m]})$.
\end{proposition}
\begin{proof}
Let us consider the diagonal prolongations $\omega^{[m]}_1,\ldots,\omega^{[m]}_k$. The differential of the diagonal prolongation of a differential form is the prolongation of the differential of the differential form (cf. \cite{CGLS13}). Hence, $\omega^{[m]}_1,\ldots,\omega^{[m]}_k$ are closed. Let us show that
$$\mathcal{D}\equiv\bigcap_{i=1}^k\ker \omega^{[m]}_i=\{0\}.$$
We define $\pi_r:N^m\rightarrow N$ to be the projection of $N^m$ onto the $r$--th component of $N^m$. If $X$ takes values in $\mathcal{D}$, then
$$
0=\omega_i^{[m]}\left(X,\frac{\partial}{\partial x^i_{(r)}}\right)=(\pi^*_r\omega_i)\left(X,\frac{\partial}{\partial x^i_{(r)}}\right),\qquad i=1,\ldots,k.
$$
Hence, $(\pi_{r})_*X\in\cap_{i=1}^k\ker(\pi_{r}^*\omega_i)=0$. So, $(\pi_{r})_*X=0$. Repeating the same for each $r$, we obtain $X=0$. Therefore, $(\omega_1^{[m]},\ldots,\omega_k^{[m]})$ is a $k$--symplectic structure.
\end{proof}

\begin{definition} Given a polysymplectic form $\Omega=\sum_{i=1}^k\omega_i\otimes e^i$ on $N$, its diagonal prolongation to $N^m$ is the polysymplectic form $\Omega^{[m]}=
\sum_{i=1}^k\omega^{[m]}_i\otimes e^i$.
\end{definition}

Let us illustrate the above notion through a remarkable example. Consider again the Schwarzian equation (\ref{Schwarz}) as a first-order system. Several works have dealt with a superposition rule for such equations \cite{CGLS13,LS13}. To obtain such a superposition rule, these works obtained three functionally independent constants of motion for the diagonal prolongation of (\ref{Schwarz}) to $\mathcal{O}_2^{[2]}$. Let us derive such constants of motion through the methods of this work in order to show the advantages of our approach.

Schwarzian equations are related to a two-symplectic structure $(\omega_1,\omega_2)$ on $\mathcal{O}_2$ given by (\ref{forms}). In view of Proposition \ref{Theo:prol}, their prolongations to $\mathcal{O}_2^{[2]}$ give rise to a two-symplectic structure on $\mathcal{O}_2^{[2]}$. Indeed, we have that the prolongations of $\omega_1$ and $\omega_2$ to $\mathcal{O}_2^{[2]}$ read
$$
\omega^{[2]}_1=\sum_{i=1}^2\frac{{\rm d}v_{(i)}\wedge {\rm d}a_{(i)}	 }{v_{(i)}},\qquad \omega^{[2]}_2=-\sum_{i=1}^2\frac 2{v_{(i)}^3}(x_{(i)}{\rm d}v_{(i)}\wedge {\rm d}a_{(i)}+v_{(i)}{\rm d}a_{(i)}\wedge {\rm d}x_{(i)}+a_{(i)}{\rm d}x_{(i)}\wedge {\rm d}v_{(i)}).
$$
Their kernels are given by
$$
\ker \omega^{[2]}_1= \left\langle \frac{\partial}{\partial x_{(1)}},\frac{\partial}{\partial x_{(2)}}\right\rangle, \qquad \ker \omega_2^{[2]}=\bigoplus_{i=1}^2\left\langle x_{(i)}\frac{\partial}{\partial x_{(i)}}+v_{(i)}\frac{\partial}{\partial v_{(i)}}+a_{(i)}\frac{\partial}{\partial a_{(i)}}\right\rangle.
$$
As proved in Proposition \ref{Theo:prol}, both kernels have zero intersection. Using (\ref{3KSHamFun}) and (\ref{3KSHamFunOmega2}), we obtain that the $k$--Hamiltonian functions for the diagonal prolongations to the vector fields (\ref{VFKS1}) to $\mathcal{O}^{[2]}_2$ read
\begin{equation*}
h^{1,[2]}=\sum_{i=1}^2\left(\frac{2}{v_{(i)}}\otimes e^{1}-\frac{4x_{(i)}}{v_{(i)}}\otimes e^2\right),\qquad h^{2,[2]}=\sum_{i=1}^2\left[\frac{a_{(i)}}{v_{(i)}^2}\otimes e^1+\left(2-\frac{2a_{(i)}x_{(i)}}{v_{(i)}^2}\right)\otimes e^2\right]
\end{equation*}
and
\begin{equation*}
h^{3,[2]}=\sum_{i=1}^2\left[\frac{a_{(i)}^2}{2v_{(i)}^3}\otimes e^1+\left(\frac{2a_{(i)}}{v_{(i)}}-\frac{a_{(i)}^2x_{(i)}}{v_{(i)}^3}\right)\otimes e^2\right].
\end{equation*}
It follows that
$$
\left\{h^{1,[2]},h^{2,[2]}\right\}_{\Omega^{[2]}}\!=h^{1,[2]},\qquad\qquad \left\{h^{1,[2]},h^{3,[2]}\right\}_{\Omega^{[2]}}\!=2h^{2,[2]},\qquad\qquad \left\{h^{2,[2]},h^{3,[2]}\right\}_{\Omega^{[2]}}\!=h^{3,[2]}.
$$
So, these functions close a Lie algebra isomorphic to $\mathfrak{sl}(2,\mathbb{R})$. Next, we will use the derived algebras to obtain several $t$-independent constants of motion for these systems.

We can induce from $\Omega^{[2]}$ several presymplectic structures $\Omega^{[2]}_\xi$ contracting $\Omega^{[2]}$ with an element of $\xi\in(\mathbb{R}^2)^*$, i.e. $\Omega^{[2]}_\xi=\langle \Omega^{[2]},\xi\rangle$. For instance, let $\{\theta_1,\theta_2\}$ be the dual basis to $\{e^1,e^2\}$. We therefore have the presymplectic forms
$$
\Omega_{\xi_1}\equiv\langle\Omega^{[2]},\theta_1\rangle=\omega^{[2]}_1, \qquad \Omega_{\xi_2}\equiv\langle \Omega^{[2]},\theta_2\rangle =\omega^{[2]}_2.
$$
From Proposition \ref{algebra}, the Hamiltonian functions $(h^{1,[2]})_\xi,(h^{2,[2]})_\xi,(h^{3,[2]})_\xi$, for every $\xi\in(\mathbb{R}^2)^*$, span a real Lie algebra $\mathfrak{W}$ such that $\mathfrak{sl}(2,\mathbb{R})$ is a Lie algebra extension. Since $\mathfrak{sl}(2,\mathbb{R})$ is simple, $\mathfrak{W}$ is isomorphic to $\mathfrak{sl}(2,\mathbb{R})$ or zero.

If the Lie algebra is isomorphic to $\mathfrak{sl}(2,\mathbb{R})$, it was proved in \cite{CGLS13} that  $\{C_\xi,(h_ i)_\xi\}_\xi=0$, where $i=1,2,3$, $\{\cdot,\cdot\}_\xi$ is the Poisson bracket on the space of admissible functions of $\Omega^{[2]}_\xi$ and
$$
C_\xi=(h^{1,[2]})_\xi (h^{3,[2]})_\xi-(h^{2,[2]})_\xi^2.
$$
It is relevant that $C_\xi$ can be obtained from a Casimir element of a Lie algebra isomorphic to $\mathfrak{sl}(2,\mathbb{R})$ constructed induced by $h^{1,[2]},h^{2,[2]},h^{3,[2]}$.  Observe that $C_\xi$ is a $t$-independent constant of motion for the prolongated system $\widetilde{X}^{[2]}_{3KS}$. More generally, a similar procedure can be developed for other Lie algebras of functions associated to $k$--symplectic Lie systems.

If we write $\xi=\lambda_1\theta_1+\lambda_2\theta_2$, with $\lambda_1,\lambda_2\in\mathbb{R}$, we have that
$$
C_\xi=[\lambda_1(h^{1,[2]})_{\xi_1}+\lambda_2(h^{1,[2]})_{\xi_2}] [\lambda_1(h^{3,[2]})_{\xi_1}+\lambda_2(h^{3,[2]})_{\xi_2}]-[\lambda_1(h^{2,[2]})_{\xi_1}+\lambda_2(h^{2,[2]})_{\xi_2}]^2
$$
and
\begin{multline*}
C_\xi=\lambda_1^2[(h^{1,[2]})_{\xi_1}(h^{3,[2]})_{\xi_1}-(h^{2,[2]})_{\xi_1}^2]+\lambda_2^2[(h^{1,[2]})_{\xi_2}(h^{3,[2]})_{\xi_2}-(h^{2,[2]})_{\xi_2}^2]+\\\lambda_1\lambda_2[(h^{1,[2]})_{\xi_1}(h^{3,[2]})_{\xi_2}+(h^{3,[2]})_{\xi_1}(h^{1,[2]})_{\xi_2}-2(h^{2,[2]})_{\xi_2}(h^{2,[2]})_{\xi_1}].
\end{multline*}
So, we can write $C_\xi=\lambda_1^2C_{\xi_1}+\lambda_2^2C_{\xi_2}+\lambda_1\lambda_2F_{\xi_1\xi_2}$ where $C_{\xi_1}$, $C_{\xi_2}$ and $F_{\xi_1\xi_2}$ are three constants of motion given by
$$
C_{\xi_1}=(h^{1,[2]})_{\xi_1}(h^{3,[2]})_{\xi_1}-(h^{2,[2]})_{\xi_1}^2=\frac{(a_2 v_1 - a_1 v_2)^2}{v_1^3 v_2^3},
$$
$$
\qquad
C_{\xi_2}=(h^{1,[2]})_{\xi_2}(h^{3,[2]})_{\xi_2}-(h^{2,[2]})_{\xi_2}^2=-4\left(-x_1x_2+\frac{2v_1v_2(v_1x_2-v_2x_1)}{a_1v_2-v_1a_2}\right)\frac{(a_2 v_1 - a_1 v_2)^2}{v_1^3 v_2^3}-4^2,
$$
\begin{multline*}
F_{\xi_1\xi_2}=(h^{1,[2]})_{\xi_1}(h^{3,[2]})_{\xi_2}+(h^{3,[2]})_{\xi_1}(h^{1,[2]})_{\xi_2}-2(h^{2,[2]})_{\xi_2}(h^{2,[2]})_{\xi_1}\\=-\frac{2(a_2v_1-v_2a_1)^2}{v_1^3v_2^3}\left(x_1+x_2-\frac{2v_1v_2(v_1-v_2)}{a_1v_2-v_1a_2}\right).
\end{multline*}

Using that $C_{\xi_1}$ is a $t$-independent constant of motion, $C_{\xi_2},F_{\xi_1\xi_2}$ allow us to define three simpler $t$-independent constants of motion $F_1,F_3,F_4$:
$$\begin{gathered}
F_1=x_1x_2-\frac{2v_1v_2(v_1x_2-v_2x_1)}{a_1v_2-v_1a_2},\qquad F_3=x_1+x_2-\frac{2v_1v_2(v_1-v_2)}{a_1v_2-v_1a_2},\\ F_4=\sqrt{F_3^2-4F_1+\frac{16 }{C_{\xi_1}}}=x_1-x_2-\frac{2v_1v_2(v_1+v_2)}{a_1v_2-v_1a_2}.
\end{gathered}
$$
The $t$-independent constants of motion $C_{\xi_2},F_3$ and $F_4$ are the first-integrals employed in \cite{CGLS13,LS13} to obtain the superposition rule for Schwarzian equations in first-order form. In those works, $C_{\xi_2},F_3,F_4$ were obtained by means of several geometric methods. In \cite{LS13} they were derived by means of the method of characteristics, which is quite long and tedious. In \cite{CLS13}, the techniques for Dirac--Lie systems enabled us to obtain $F_1$ and $C_{\xi_2}$. Meanwhile, $F_4$ had to be obtained through a Lie symmetry. In this work, $C_{\xi_2},F_3,F_4$ appear simultaneously from the $k$--symplectic structure of Schwarzian equations. This is the key point of the usefulness of this approach to obtain superposition rules.  The $k$--symplectic structure provides a framework to exploit the geometric properties of $k$--symplectic Lie system better than Dirac--Lie systems.

\section{Conclusions and Outlook}\label{Outlook}
We have described the main properties of a new type of Lie systems, the $k$--symplectic Lie systems. This has led to describe  new Poisson structures related
to $k$--symplectic structures as well as the description of new methods to study Lie systems, e.g. their superposition rules, constants of motion, etc. On  the advantages of the $k$--symplectic structures is that only the $k$--symplectic structures provide a geometric framework containing all the geometric structure of a $k$--symplectic Lie system. Finally, this paper opens a new setting of applications of the $k$--symplectic structures, since we use this geometrical structures for studying systems of differential equations. At present, the $k$--symplectic geometry is applied to the study of first-order classical field theories.

In the future we aim to develop a theory of momentum maps for $k$--symplectic Lie systems as well as to study the structure of restrictions of $k$--symplectic Lie systems to $k$--symplectic isotropic/coisotropic and Lagrangian $k$--symplectic submanifolds. We also aim to investigate in depth the existence of $k$--symplectic Lie systems on $\mathbb{R}^3$.

Moreover, we plan to extend our methods to the realm of the so-called PDE Lie systems: the natural generalisation of Lie systems to partial differential equations. It seems that the $k$--symplectic theory can be employed in this topic to provide a new geometric framework for the description of such systems. We hope to apply our findings to new
interesting PDE Lie systems of relevance.

\section*{Acknowledgements}

Research of J. de Lucas is partially financed by research projects MTM2010-12116-E (Ministerio de Ciencia e Innovaci\'on) and Polish National Science Centre grant HARMONIA
under the contract number DEC-2012/04/M/ST1/00523. Research of S. Vilari\~{n}o is partially financed by research projects MTM2011-15725-E and MTM2011-2585 (Ministerio de Ciencia e Innovaci\'{o}n) and E24/1 (Gobierno de Arag\'on).


\begin{thebibliography}{15}
\bibitem{FM}
R. Abraham and  J.E. Marsden,
Foundations of Mechanics. Second Edition, (Redwood City: Addison--Wesley), 1987.


\bibitem{CRC94}
R.L. Anderson, V.A. Baikov, R.K. Gazizov, W. Hereman, N.H. Ibragimov,  F.M. Mahomed, S.V. Meleshko, M.C. Nucci, P.J. Olver, M.B. Sheftel', A.V. Turbiner
and E.M. Vorob'ev,  CRC handbook of Lie group analysis of
differential equations. Vol. { 1}. Symmetries, exact solutions and conservation laws, (Boca Raton: CRC Press), 1994.

\bibitem{AHW81}
R.L. Anderson, J. Harnad  and P. Winternitz,
Group theoretical approach to superposition rules for systems of Riccati
equations, \textit{Lett. Math. Phys.} {5} (1981) 143--148.

\bibitem{ADR12}
R.M. Angelo, E.I. Duzzioni  and A.D. Ribeiro,
Integrability in $t$-dependent systems with
one degree of freedom, \textit{J. Phys. A: Math. Theor.} {45} (2012) 055101.

\bibitem{A89}
V.I. Arnold,
Mathematical methods of classical mechanics. Second edition,
(New York: Springer-Verlag), 1989.

\bibitem{Aw-1992}
A. Awane, $k$--symplectic structures, \textit{J. Math. Phys.} {33} (1992) 4046--4052.

\bibitem{BBHLS13}
A. Ballesteros, A. Blasco, F. Herranz, J.de~Lucas and  C.~S. Sard\'on,
Lie--Hamilton systems on the plane: theory, generalisations and applications,
arXiv:1311.0792.


\bibitem{BCHLS13}
A. Ballesteros, J.~F. Cari\~nena, F. Herranz, J. de~Lucas and C.~S. Sard\'on,
From constants of motion to superposition rules for Lie--Hamilton
  systems.
\newblock {J. Phys. A: Math. Theor.}{  46 } (2013) 285203.


\bibitem{Be07}
L.M. Berkovich,
Method of factorization of ordinary differential operators and some of its
applications,
\textit{Appl. Anal. Discret. Math.} {1} (2007) 122--149.


\bibitem{Be84}
M.V. Berry,
Quantal phase factors accompanying adiabatic changes,
\textit{Proc. R. Soc. Lond. Ser.  A} { 392} (1984) 45--57.

\bibitem{CCR03}
J. F. Cari\~{n}ena,  J. Clemente-Gallardo and A. Ramos,
 Motion on Lie groups and its applications in control theory, Proceedings of the XXXIV Symposium on Mathematical Physics (Torun, 2002). \textit{Rep. Math. Phys.} {51} no. 2-3  (2003) 159--170.

\bibitem{CGL09}
J.F. Cari\~nena, J. Grabowski and J. de~Lucas,
Quasi-Lie schemes: theory and applications, \textit{J. Phys. A: Math. Theor.} {42} (2009) 335206.


\bibitem{CGL11}
J.F. Cari\~nena, J. Grabowski and J. de~Lucas,
Superposition rules for higher-order differential equations, and their
applications,
\textit{J. Phys. A: Math. Theor.} {45} (2012) 185202.

\bibitem{CGLS13}
J.F. Cari\~nena, J. Grabowski, J. de~Lucas and C. Sard\'{o}n,
Dirac--Lie systems and Schwarzian equations,
\textit{J. Differential equations} {257} (2014) 2303--2340.

\bibitem{CGM00}
J.F. Cari\~nena, J. Grabowski and G. Marmo,
Lie--Scheffers systems: a geometric approach. (Naples: Bibliopolis), 2000.



\bibitem{CGM07}
J.F. Cari\~nena, J. Grabowski and G. Marmo,
Superposition rules, Lie theorem and partial differential equations,
\textit{Rep. Math. Phys.} {60} (2007) 237--258.

\bibitem{CarRamGra}
J.F. Cari\~nena, J. Grabowski and A. Ramos,
Reduction of time-dependent systems admitting a superposition principle,
\textit{Acta Appl. Math.} {66} (2001) 67--87.

\bibitem{CGR09}
J.F. Cari\~nena, P. Guha and M.F. Ra\~nada,
A geometric approach to higher-order Riccati chain: Darboux polynomials and
constants of the motion,
\textit{J. Phys.: Conf. Ser.} {175} (2009) 012009.

\bibitem{CLLEmd}
J.F. Cari\~nena, P.G.L. Leach and J. de~Lucas,
Quasi-Lie schemes and Emden--Fowler equations,
\textit{J. Math. Phys.} {50} (2009) 103515.

\bibitem{Dissertations}
J.F. Cari\~nena and J. de~Lucas,
{\it Lie systems: theory, generalisations, and applications}.
Diss. Math. (Rozprawy Math.)  {479}. Warsaw: Institute of Mathematics of the Polish Academy of Sciences, 2011.

\bibitem{IntRic11}
J.F. Cari\~nena and J. de Lucas,
{\it Integrability of Lie systems through Riccati equations},
\textit{J. Nonlinear Math. Phys.} 18 (2011) 29--54.

\bibitem{CLR08}
J.F. Cari\~nena, J. de~Lucas and M.F. Ra\~nada,
Recent applications of the theory of Lie systems in Ermakov systems,
\textit{SIGMA Symmetry Integrability Geom. Methods Appl.} {4} (2008) 031.

\bibitem{CLS13}
J.F. Cari\~nena, J. de~Lucas and C. Sard\'{o}n,
Lie--Hamilton systems: theory and applications,
\textit{Int. J. Geom. Methods Mod. Phys.} { 10} (2013) 1350047.




\bibitem{CV13}
J.N. Clelland  and P.J. Vassiliou,
A solvable string on a Lorentzian surface,
\textit{Differential Geom. Appl.} { 33} (2014) 177--198.


\bibitem{Clem06}
J. Clemente-Gallardo,
On the relations between control systems and Lie systems, in: \textit{Groups,
geometry and physics. Monogr. Real Acad. Ci. Exact. F\'is.-Qu\'im. Nat.
Zaragoza} {29}. (Zaragoza: Acad. Cienc. Exact. F\'is. Qu\'im. Nat. Zaragoza), 2006, pp.  65--68.

\bibitem{Co90}
T.J. Courant,
Dirac manifolds,
{\it Trans. Amer. Math. Soc.} { 319} (1990) 631--661.

\bibitem{Darboux}
G. Darboux,
Sur la th\'{e}orie des coordinne\'{e}s cuvilignes et les syst\'{e}mes orthogonaux.
\textit{Ann. Ec. Norm. Sup\'{e}r.} {7} (1878) 101--150.

% \bibitem{De08}
% V.P. Dereveskii,
% Matrix Bernuilli equations,
% \textit{Russ. Math.} { 52} (2008) 1--7

\bibitem{Ru08}
R. Flores-Espinoza,
Monodromy factorization for periodic Lie systems and reconstruction
phases, in: %Kielanowski, P., Odzijewicz, A., Schlichenmaier, M. (eds.)
{\it Geometric methods in Physics}. AIP Conf. Proc. { 1079}. (New York: Amer. Inst. Phys.), 2008, pp. 189--195.

\bibitem{Ru10}
R. Flores-Espinoza,
Periodic first integrals for Hamiltonian systems of Lie type,
\textit{Int. J. Geom. Methods Mod. Phys.} {8} (2011) 1169--1177.


\bibitem{GKO92}
A. Gonz\'{a}lez-L\'{o}pez, N. Kamran and P.J. Olver,
 Lie algebras of vector fields in the real
plane, \textit{Proc. London Math. Soc.} {64} (1992) 339--368.



\bibitem{GL13}
J. Grabowski and J. de Lucas,
Mixed superposition rules and the Riccati hierarchy.
\textit{J. Differential Equations} {254} (2013) 179--198.

\bibitem{Gu-1987}
C. G\"{u}nther,
The polysymplectic Hamiltonian formalism in field theory and calculus of variations I: The local case,
\textit{J. Differential Geom.} {25} (1987)  23--53.

\bibitem{Ib00}
N. H. Ibragimov,
{Discussion of Lie's nonlinear superposition theory}, in: {\sl Proceedings of the International
Conference of Modern Group Analysis for the New Millennium (MOGRAN '00)}. (Russia: Ufa), 2000.

\bibitem{Ib09}
N. H. Ibragimov,
Utilization of canonical variables for integration of systems of first-order
differential equations,
\textit{ALGA} {6} (2009) 1--18.

\bibitem{JT12}
M. Jotz and T. S. Ratiu,
Dirac structures, nonholonomic systems and reduction
\textit{Rep. Math. Phys.} {69} (2012) 5--56.

\bibitem{KKW08}
V. Kravchenko, V. Kravchenko and B. Williams,
A quaternionic generalisation of the Riccati differential equation, in: {\sl Clifford Analysis and Its Application} (Dordrecht: Kluwer Acad. Publ.),  2001, pp. 143--154.
%(2008)

\bibitem{LLS11}
N. Lanfear, R.M. L\'opez  and  S.K. Suslov,
Exact wave functions for generalized harmonic oscillators,
\textit{J. Russ. Laser Research} { 32} (2011) 352--361.

\bibitem{LMS-1988}
M. de Le\'{o}n, I. M\'{e}ndez and M. Salgado,
Regular $p$-almost cotangent structures,
\textit{J. Korean Math. Soc.} {25} (1988) 273--287.

\bibitem{LMS-1993}
M. de Le\'{o}n, I. M\'{e}ndez and M. Salgado,
$p$-almost cotangent structures,
\textit{Boll. U.M.I}. {A7(7)} (1) (1993) 97--107.

\bibitem{LV-2012}
M. de Le\'{o}n and S. Vilari\~{n}o,
Lagrangian submanifolds in $k$--symplectic settings, \textit{Monatshefte f\"{u}r Mathematik} {170} (3-4) (2013) 381--404.



\bibitem{LM87}
P. Libermann  and C.-M. Marle,
Symplectic Geometry and Analytical Mechanics, (Dordrecht: D. Reidel Publishing Co.), 1987.

\bibitem{LS}
S. Lie  and G. Scheffers,
Vorlesungen \" uber continuierliche Gruppen mit geometrischen und
  anderen Anwendungen,
(Leipzig: Teubner), 1893.

\bibitem{LZ11}
Si-Qi Liu and  Y. Zhang,
Jacobi structures of evolutionary partial differential equations,
{\it Advances in Mathematics}
{ 227} (2011) 73--130.


\bibitem{LS13}
 J. de~Lucas   and C. Sard\'{o}n,
On Lie systems and Kummer--Schwarz equations,
\textit{J. Math. Phys.} { 54} (2013) 033505.

\bibitem{MRSV13}
J.C. Marrero,  N. Rom\'an-Roy, M. Salgado  and S. Vilari\~no,
Reduction of polysymplectic manifolds,
ArXiv:1306.0337v2.

\bibitem{merino}
E. Merino,
Geometr\'{\i}a $k$--simpl\'{e}ctica y $k$--cosimpl\'{e}ctica. Aplicaciones a las teor\'{\i}as cl\'{a}sicas de campos,
(Publicaciones del Dtp de Geometr\'{\i}a y Topolog\'{\i}a 87. Universidad de Santiago de Compostela), 1997.

\bibitem{MRS04}
F. Munteanu, A.M. Rey and M. Salgado,
    The G\"{u}nther's formalism in classical field theory: momentum map and reduction,
    \textit{J. Math. Phys.}, { 45} (2004) 1730--1751.

\bibitem{NM13}
J.C. Ndogmo   and F.M. Mahomed,
On certain properties of linear iterative equations.
\textit{Center European J. Math.} { 56}  (2013) 34--36.

\bibitem{Ni00}
S. Nikitin,
Control synthesis for \u{C}aplygin polynomial systems,
\textit{Acta Appl. Math.} { 60}  (2000) 199--212.

\bibitem{Ol91}
P.J. Olver,
Applications of Lie groups to differential equations, Graduate Texts in
Mathematics { 107}, (New York: Springer-Verlag), 1993.

\bibitem{JPOT}
J.P. Ortega  and T.S. Ratiu,
Momentum maps and Hamiltonian reduction, Prog. Math. { 222},
 (Boston: Birkh\"auser Boston, Inc.), 2004.

\bibitem{OT09}
V. Ovsienko and S. Tabachnikov,
What is the {S}chwarzian derivative,	
\textit{Notices of the AMS}\, { 56} (2009) 34--36.

\bibitem{Palais}
R.S. Palais,
Global formulation of the Lie theory of transformation groups,
{\it Mem. Amer. Math. Soc.} {22} (1957) 1--123.

\bibitem{Pi13}
G. Pietrzkowski,
 Explicit solutions of the $\mathfrak{a}_1$-type Lie--Scheffers system and a general Riccati equation,
\textit{J. Dyn. and Control Sys.} { 18}  (2012) 551--571.

\bibitem{Ra06}
A. Ramos,
New links and reductions between the brockett nonholonomic integrator and related systems,
\textit{Rend. Semin. Mat. Univ. Politec. Torino} { 64} (2006) 39--54.




\bibitem{WintSecond}
C. Rogers, W.K. Schief  and P. Winternitz,
Lie-theoretical generalisation and discretization of the Pinney Equation,
\textit{J. Math. Anal. Appl.} {216} (1997) 246--264.

\bibitem{RSV07}
N. Rom\'an-Roy,  M. Salgado and S. Vilari\~no,
Symmetries and conservation laws in the Gunther $k$--symplectic formalism of field theory,
\textit{ Rev. Math. Phys.}  { 19}(10) (2007) 1117--1147.

\bibitem{RA}
D.E. Rourke   and M.P. Augustine,
{Exact linearization of the radiation-damped spin system},
\textit{Phys. Rev. Lett. }{84} (2000) 1685--1688.

% \bibitem{SW08}
% M. Sorine  and P. Winternitz,
% Superposition laws for solutions of differential matrix Riccati equations arising in control theory,
% \textit{IEEE Trans. Automat. Control} { 30} (1985) 266--272.

\bibitem{SSV11}
 E. Suazo, S.K. Suslov  and J. M. Vega-Guzm\'an,
 The Riccati equation and a diffusion-type equation,
\textit{ New York J. Math.} { 17} (2011) 225--244.

\bibitem{IV}
I. Vaisman,
Lectures on the geometry of Poisson Manifolds,
(Basel: Birkh\"auser Verlag), 1994.

\bibitem{Va13}
P.J. Vassilou,
Cauchy problem for a Darboux integrable wave map system and equations of Lie type,
\textit{SIGMA} { 9}  (2013) 024.

\bibitem{WTC83}
J. Weiss,  M. Tabor and G. Carnevale,
The Painlev\'e property for partial differential equations,
\textit{J. Math. Phys.} {24}  (1983) 522--526.

\bibitem{Wi09}
P. Wilczy\'nski,   Quaternionic-valued ordinary differential equations, The Riccati equation,
\textit{ J. Differential Equations} {247} (2009) 2163--2187.

\bibitem{PW}
P. Winternitz,
Lie groups and solutions of nonlinear differential equations, in:
{ Nonlinear phenomena. Lect, Notes Phys.} { 189},
(Oaxtepec: Springer-Verlag), 1983, pp. 263--331.

\bibitem{YM06}
H. Yoshimura and J. E. Marsden,
Dirac Structures in Lagrangian Mechanics. Part I: Implicit Lagrangian Systems,
{\it J. Geom. Phys.} { 57} (2006), 133--156.


\end{thebibliography}
\end{document}